\documentclass[11pt]{article}

\usepackage[letterpaper,margin=1in]{geometry}
\usepackage{subfig}%
\usepackage{bbm}
\usepackage{amsmath,amssymb,amsfonts}
\usepackage[boxruled, linesnumbered]{algorithm2e}
\SetKwInput{KwInput}{Input}                %
\SetKwInput{KwOutput}{Output} 
\DontPrintSemicolon
\usepackage{hyperref}
\hypersetup{colorlinks=true, linkcolor=blue!80!black, urlcolor=blue!80!black, citecolor=blue!80!black}
\usepackage[nameinlink]{cleveref}

\crefname{assumption}{assumption}{assumptions}
\Crefname{assumption}{Assumption}{Assumptions}

\crefrangelabelformat{assumption}{#3#1#4--#5#2#6}

\usepackage{lscape}
\usepackage{graphicx}
\usepackage{tikz}
\usepackage{algorithmic}
\usetikzlibrary{trees}
\usepackage{pgfplots}
\usepackage{adjustbox}
\usepackage{multicol}
\usepackage{multirow}
\usepackage{mathtools}
\usepackage{float}
\usepackage{comment}
\usepackage[square, comma, sort&compress, numbers]{natbib}
\usetikzlibrary{shapes,positioning,fit,calc}

\def\argmax{\mathop{\rm arg\,max}}%

\pgfplotsset{compat=1.13}

\usepackage{colortbl}
\usepackage{pdflscape}
\usepackage{xcolor}
\usepackage{cite,doi}               %
\usepackage{bbold}

\newcolumntype{R}{>{\columncolor{red!20}}c}
\newcolumntype{G}{>{\columncolor{green!20}}c}
\newcolumntype{B}{>{\columncolor{blue!20}}c}
\newcolumntype{Y}{>{\columncolor{yellow!20}}c}
\newcolumntype{K}{>{\columncolor{black!20}}c}

\usepackage{courier}
\usepackage{listings}
\usepackage[toc,page]{appendix}

\lstdefinestyle{C++}{
	language=C++,
	keywordstyle=\bfseries\color{purple},
	stringstyle=\color{blue},
	commentstyle=\color{gray},
	comment=[l]{/*}
}

\lstdefinestyle{AMPL}{
	language=AMPL,
	aboveskip=3mm,
	belowskip=3mm,
	showstringspaces=false,
	columns=flexible,
	keywordstyle=\bfseries,
	breaklines=true,
	breakatwhitespace=true,
	tabsize=3,
}

\lstdefinelanguage{AMPL}{keywords={let,set,param,var,arc,integer,minimize,maximize,subject,to,node,
sum,in,Current,complements,integer,solve_result_num,IN,contains,less,suffix,INOUT,default,logical,
Infinity,dimen,max,symbolic,Initial,div,min,table,LOCAL,else,option,then,OUT,environ,setof,union,
all,exists,shell_exitcodeuntil,binary,forall,solve_exitcodewhile,by,if,solve_messagewithin,check,
solve_result},sensitive=true,comment=[l]{\#}}

\newcommand{\bvec}{\left(\begin{array}{c}}
\newcommand{\evec}{\end{array}\right)}
\newcommand{\bsub}{\begin{subequations}}
\newcommand{\esub}{\end{subequations}}

\newcommand{\FUB}{F_{\mathrm{UB}}}

\usepackage{longtable}

\SetKwFunction{FRecurs}{FnRecursive}%
\SetKwProg{Function}{}{}{}\SetKwFunction{FRecurs}{void FnRecursive}%
\SetKwComment{Comment}{/* }{ */}

\usepackage{environ}

\definecolor{peach}{HTML}{FFCCAC}
\definecolor{butter}{HTML}{FFEB94}
\definecolor{babyblue}{HTML}{C1E1DC}

\NewEnviron{globalization_mechanism}{%
	\fcolorbox{peach}{peach}{%
	\begin{minipage}[l]{0.88\columnwidth}
	\BODY
	\end{minipage}
	}
}

\NewEnviron{globalization_strategy}{%
	\fcolorbox{butter}{butter}{%
	\begin{minipage}[l]{0.9\columnwidth}
	\BODY
	\end{minipage}
	}
}
\NewEnviron{step_computation}{%
	\fcolorbox{babyblue}{babyblue}{%
	\begin{minipage}[l]{0.9\columnwidth}
	\BODY
	\end{minipage}
	}
}

\usepackage{amsthm}

\newtheorem{theorem}{Theorem}

\newtheorem{proposition}{Proposition}

\newtheorem{remark}{Remark}

\newtheorem{assumption}{Assumption}

\usepackage{authblk}
\begin{document}

\title{Binary Quantum Control Optimization with Uncertain Hamiltonians}
\author[a]{Xinyu Fei}
\author[b]{Lucas T. Brady}
\author[c]{Jeffrey Larson}
\author[c]{Sven Leyffer}
\author[a]{Siqian Shen}

\affil[a]{Industrial and Operations Engineering Department, University of Michigan. {\tt xinyuf@umich.edu, siqian@umich.edu}}
\affil[b]{KBR @ NASA Ames Quantum Artificial Intelligence Laboratory. {\tt lucas.t.brady@nasa.gov}}
\affil[c]{Mathematics and Computer Science Division, Argonne National Laboratory.  {\tt jmlarson@anl.gov, leyffer@anl.gov}}
\date{}

\maketitle

\begin{abstract}
Optimizing the controls of quantum systems plays a crucial role in advancing quantum technologies. 
The time-varying noises in quantum systems and the widespread use of inhomogeneous quantum ensembles raise the need for high-quality quantum controls under uncertainties. 
In this paper, we consider a stochastic discrete optimization formulation of a binary optimal quantum control problem involving Hamiltonians with predictable uncertainties.
We propose a sample-based reformulation that optimizes both risk-neutral and
risk-averse measurements of control policies, and solve these with two
gradient-based algorithms using sum-up-rounding approaches. 
Furthermore, we discuss the differentiability of the objective function and
prove upper bounds of the gaps between the optimal solutions to binary control
problems and their continuous relaxations.
We conduct numerical studies on various sized problem instances based of two
applications of quantum pulse optimization; we evaluate different strategies to
mitigate the impact of uncertainties in quantum systems. 
We demonstrate that the controls of our stochastic optimization model achieve
significantly higher quality and robustness compared to the controls of a
deterministic model.
\end{abstract}

\section{Introduction}
Quantum control~\citep{Glaser_2015,Werschnik_2007,d2021introduction} focuses on designing efficient and accurate controls that manipulate quantum systems toward desired quantum states and operations.
Applications of quantum control include nuclear magnetic resonance experiments~\citep{kehlet2004improving,khaneja2005optimal,maximov2008optimal,skinner2003application} and quantum chemistry~\citep{Kosloff_1989,Peirce_1988}.
With the development of quantum technologies, quantum control plays an important role in quantum information~\citep{pawela2016various,Motzoi2009,gokhale2019partial,leung2017speedup,Palao2002,Palao2003,rebentrost2009optimal} and the high-level design of quantum algorithms~\citep{Brady_2021,brady2021annealing,Fei2023binarycontrolpulse,fei2023switching,Bapat_2018,Mbeng_2019,Venuti_2021}.

Various methods have been developed to solve quantum optimal control problems.
\citet{khaneja2005optimal} develop the gradient ascent pulse engineering (GRAPE) algorithm, which estimates controls by piecewise constant functions.
\citet{Larocca_2021} improve the computational efficiency of the GRAPE algorithm using a Krylov subspace approach.
Another popular method is the chopped random basis algorithm, which describes the control space using basis functions~\citep{doria2011optimal,caneva2011chopped,sorensen2018quantum}.
Other studies solve these optimal control problems using gradient-free methods, including evolution algorithms \citep{zahedinejad2014evolutionary} and reinforcement learning approaches \citep{bukov2018reinforcement,niu2019universal,sivak2022model,chen2013fidelity,zhang2019does}.
For binary control problems, \citet{vogt2021binary} propose a trust-region method for the optimal control of a single-flux quantum system.
\citet{Fei2023binarycontrolpulse} develop a solution framework for general quantum systems and further improve it using a switching time approach \citep{fei2023switching}.
 
The aforementioned papers study only deterministic quantum optimal control problems, which have fixed Hamiltonians, Hermitian operators that generate the time evolution of the system and whose eigenvalues correspond to energy levels of the system.  These Hamiltonians are specified by the experimental setup with the control coming from how long we use each Hamiltonian and in what order for time evolving the state of the system.
However, the imprecise estimation of the Hamiltonian controllers and time-varying noises in quantum systems has recently raised the need for robust quantum control \citep{dahleh1990optimal,grace2012optimized,propson2022robust,green2013arbitrary,kosut2022robust}.
Moreover, designing a robust uniform control is an important topic in inhomogeneous quantum ensembles, involving a large number of quantum systems with variations in system parameters \citep{chen2014sampling,li2006control,pryor2006fourier,owrutsky2012control,mischuck2012control}.
Fourier decomposition methods can be applied to design a uniform control for inhomogeneous quantum fields \citep{li2006control,pryor2006fourier}.
\citet{barr2022qubit} extend quantum noise spectroscopy to design optimized amplitude control waveforms that suppress low-frequency dephasing noise and detuning errors.
\citet{ruths2011multidimensional} propose a multidimensional pseudospectral method with uncertainty sampling for optimal control of quantum ensembles.
\citet{chen2014sampling} apply a sample approximation algorithm, and \citet{wu2019learning} extend it to a batch-based sampling algorithm that minimizes the expected error between final and target operations.
To hedge against risk, other studies focus on optimizing the worst-case performance under uncertainties.
\citet{wesenberg2004designing} solves a robust quantum optimal control problem using a general minmax algorithm based on a series of constrained quasi-Newton sequential quadratic programs.
\citet{kosut2013robust} develop a sample-based sequential convex programming scheme to obtain an optimal control for the worst-case robust optimization problem.

Here we seek to make three major contributions to the literature. First, we develop a stochastic optimization model and a sample-based reformulation for the general quantum optimal control problem under uncertain Hamiltonians; this model balances risk-neutral and risk-averse objectives.
Second, we apply multiple gradient-based methods and rounding techniques to solve the reformulated mixed-integer stochastic programming model.
We provide the derivative of the objective function as well as derive bounds for the gap between solutions before and after rounding.
Third, we analyze the performance of our approaches under various variance settings and demonstrate the benefits of considering uncertainties when conducting binary controls of various quantum systems.

The manuscript is organized as follows. \Cref{sec:model} presents a general mixed-integer stochastic optimization model and its reformulation based on finite samples of the uncertain Hamiltonians.
In \Cref{sec:algorithm} we derive the reformulation of the original stochastic optimization model and propose our gradient-based algorithm to solve the continuous relaxation. We apply rounding techniques to obtain binary controls and analyze the gap between binary and continuous control solutions.
\Cref{sec:numeric} introduces two specific quantum control instances and discusses the results of our numerical tests and simulation.
\Cref{sec:conclusion} summarizes our study and states future research directions.

\section{Modeling Uncertain Hamiltonians}
\label{sec:model}
For completeness, \Cref{sec:model-deter} first introduces the model for
deterministic binary quantum optimal control presented in \citet{fei2023switching}.
In \Cref{sec:model-sp}, we extend this model to a general stochastic optimization model and we propose a sample-based reformulation.
\Cref{sec:risk-measure} describes how the objective function can be modified to account for various risk measures.

\subsection{Deterministic Optimal Control Model}
\label{sec:model-deter}
Consider a quantum system with $q$ qubits, and let the control process be conducted on the time interval $[0,t_f]$, where $t_f$ is defined as the evolution time.
Let $H^{(0)}\in \mathbb{C}^{2^q\times 2^q}$ be the intrinsic Hamiltonian, a Hamiltonian over which we do not have control and which is always applied in the time evolution of the system.
Let $N$ be the number of control Hamiltonians
given by $H^{(j)}\in \mathbb{C}^{2^q\times 2^q},\ j=1,\cdots,N$.
Let $X_{\mathrm{init}},\ X_{\mathrm{targ}}\in \mathbb{C}^{2^q\times 2^q}$ be the initial and target unitary operators of the quantum system, respectively.
We define $T$ as the number of time steps and divide the time horizon $[0,t_f]$ into $T$ time intervals $(t_{k-1}, t_{k}],\ k=1,\cdots,T$.
(In this work we use a uniform time discretization where each time interval has an equal length $\Delta t=t_f/T$, but our work can be extended to discretization with nonuniform time interval length.)
For each controller $j=1,\cdots,N$ and each time step $k=1,\cdots,T$, we define discretized control variables as $u_{jk}$.
For each time step $k=1,\cdots,T$, we define the discretized time-dependent Hamiltonians as $H_{k}\in \mathbb{C}^{2^q\times 2^q}$ and unitary operators as $X_k\in \mathbb{C}^{2^q\times 2^q}$.

We define the general deterministic binary quantum control model as follows.
\begin{subequations}
\makeatletter
\def\@currentlabel{D}
\makeatother
\label{eq:model-d-1}
\begin{align}
    \label{eq:model-d-1-obj}
    (\mathrm{D})\quad \min_{u,X,H} \quad & F_X(X_T)\\
    \label{eq:model-d-1-cons-h}
    \mathrm{s.t.}\quad & H_k = H^{(0)} + \sum_{j=1}^N u_{jk}H^{(j)},\ k=1,\ldots,T\\
    \label{eq:model-d-1-cons-s}
    & X_{k}=e^{-i H_k \Delta t}X_{k-1},\ k=1,\ldots,T \\
    \label{eq:model-d-1-cons-i}
    & X_0 = X_\mathrm{init}\\
    \label{eq:model-d-1-cons-u-sum}
    & \sum_{j=1}^N u_{jk} = 1,\ k=1,\ldots,T\\
    \label{eq:model-d-1-cons-u}
    & u_{jk}\in \left\{0,1\right\},\ j=1,\ldots,N,\ k=1,\ldots,T.
\end{align}
\end{subequations}
The objective function $F_X$ is a general function to measure the difference between the final and desired quantum systems and can take different forms depending on specific problem instances.
(In the numerical simulations of this manuscript, we employ two widely used objective functions in the quantum control field, but our methods described later are general and can be used for general differentiable functions $F_X$.)

In quantum mechanics, the quantum state is defined as the description of the physical properties of a quantum system.
We use $|\cdot\rangle$ and $\langle\cdot|$ to represent a quantum state
vector and its conjugate transpose, respectively.  This notation is standard Dirac notation in quantum mechanics; and the state of a quantum system, $|\cdot\rangle$, can be represented numerically as a $2^q$ dimensional, complex vector, normalized to one.  We also use $\cdot^\dagger$ as a modifier to denote the conjugate transpose of a complex matrix.
We define a parameter $\tilde{H}$ as the constant Hamiltonian that determines the energy structure of the quantum system.

One function we will use later evaluates the difference between the energy of the quantum system with final operator $X_T$ and the minimum energy corresponding to $\tilde{H}$, given by
\begin{align}
    \label{eq:obj-energy}
    F_X(X_T) = 1 - \langle \psi_0 | X^\dagger_T \tilde{H} X_T |\psi_0\rangle / E_{\mathrm{min}},
\end{align}
where $|\psi_0\rangle$ represents the initial state of the quantum system and the constant minimum energy $E_\mathrm{min}$ represents the minimum eigenvalue of $\tilde{H}$, with an assumption that $E_\mathrm{min}\neq0$.

Another function we consider is the infidelity function
\begin{align}
\label{eq:obj-fid}
    F_X(X_T) = 1 - \frac{1}{2^q} \left|\textbf{tr} \left\{ X^\dagger_\mathrm{targ} X_T\right\}\right|,
\end{align}
which measures the difference between $X_T$ and the target operator $X_\mathrm{targ}$. Both objective functions \eqref{eq:obj-energy} and \eqref{eq:obj-fid} are bounded between $[0,1]$.

Constraints \eqref{eq:model-d-1-cons-h} compute time-dependent Hamiltonians $H_k$ for $k=1,\ldots,T$ as linear combinations of the intrinsic Hamiltonian and the control Hamiltonians weighted by control variables $u$.
The constraints, \eqref{eq:model-d-1-cons-s}, describe the time evolution process for computing unitary operators, $X$, by solving the Schr\"odinger equation.
Constraint \eqref{eq:model-d-1-cons-i} is the initial condition of the unitary operators.
Constraint \eqref{eq:model-d-1-cons-u-sum} ensures that at each time step, the summation of all the control values should be one, which is Special Ordered Set of Type 1 (SOS1) property in optimal control theory \citep{sager2012integer}.
Combining the SOS1 property with binary constraints for control variables \eqref{eq:model-d-1-cons-u}, we ensure that only one controller is active at any time.

\subsection{Stochastic Optimal Control Model}
\label{sec:model-sp}
In practice, the intrinsic and control Hamiltonians are affected by time-dependent noise due to various reasons such as decoherence, hardware limitations, and environmental noise~\citep{lidar2013quantum,breuer2007theory,Bruzewicz_2019}.
On the other hand, multiple applications, such as inhomogeneous quantum ensembles, require applying a uniform control to manipulate quantum systems with different Hamiltonian values.
These properties and applications lead to quantum control studies that take the uncertainty of Hamiltonians into consideration.
In this manuscript we assume that the uncertainty parameters have a known distribution, and we define a measure space $(\Xi, 2^\Xi, \mathcal{P})$, where $\Xi\subseteq \mathcal{R}^{(N+1)\cdot T}$ is the sample space and $\mathcal{P}$ represents the probability distribution function.
We denote the uncertainty parameters as $\xi=\left[\xi_{jk}\right]\in \Xi$, where $\xi_{jk},\ j=0,\ldots,N,\ k=1,\ldots,T$ represents the uncertainty of the intrinsic Hamiltonian ($j=0$) and $j$th control Hamiltonian ($j\geq 1$) in the time interval $k$.
We use finite samples to approximate the distribution, and we define the set $\mathcal{S}$ as a finite set of uncertainty realizations, $\mathcal{S}=\{\xi^1,\ldots, \xi^S\}$, according to the distribution $\mathcal{P}$ such that $\xi^s$, for each $s=1,\ldots,S$, is associated with probability $p_s$ with $\sum_{s=1}^S p_s=1$.
The time-dependent Hamiltonians and unitary operators are the functions of uncertain parameters $\xi$.
For each sample $\xi^s,\ s=1,\ldots,S$, at each time step $k=1,\ldots,T$, we denote the corresponding time-dependent Hamiltonian $H_k$ and unitary operator $X_k$ as $H_k^s$ and $X_k^s$, respectively.
We define $\rho\left[F_X(X_T^1),\ldots, F_X(X_T^S)\right]$ as a risk measure function based on the uncertainty sample set $\mathcal{S}=\{1,\ldots,S\}$. The generic stochastic optimization model variant of Model \eqref{eq:model-d-1} is given by
\begin{subequations}
\makeatletter
\def\@currentlabel{SP($\mathcal{S}$)}
\makeatother
\label{eq:model-sp-sample}
\begin{align}
    \label{eq:model-sp-sample-obj}
    (\textrm{SP($\mathcal{S}$)}) \quad
    \min_{u,X,H} \quad & \rho\left[F_X(X_T^1),\ldots,F_X(X_T^S)\right]\\
    \label{eq:model-sp-sample-cons-h}
    \textrm{s.t.}\quad  & H_k^s = (1+\xi^s_{0k})H^{(0)} + \sum_{j=1}^N (1+\xi^s_{jk})u_{jk}H^{(j)},\ k=1,\ldots,T,\ s=1,\ldots,S\\
    \label{eq:model-sp-sample-cons-s}
    & X_{k}^s=e^{-i H_k^s \Delta t}X_{k-1}^s,\ k=1,\ldots,T,\ s=1,\ldots,S \\
    \label{eq:model-sp-sample-cons-i}
    & X_0^s = X_\textrm{init},\ s=1,\ldots,S\\
    \label{eq:model-sp-samplecons-u}
    & \sum_{j=1}^N u_{jk} = 1,\ u_{jk}\in \left\{0,1\right\},\ j=1,\ldots,N,\ k=1,\ldots,T.
\end{align}
\end{subequations}
The objective \eqref{eq:model-sp-sample-obj} uses the risk measure $\rho$ to evaluate the risk of having deviations from the desired cost of controlling quantum operators under uncertainty, which we detail in \Cref{sec:risk-measure}.
(In \Cref{sec:numeric} for the numerical results, we introduce specific quantum control objective functions for different examples.)
Constraints \eqref{eq:model-sp-sample-cons-h} compute the time-dependent Hamiltonians given samples $\xi^s,\ s=1,\ldots,S$.
Constraints \eqref{eq:model-sp-sample-cons-s}--\eqref{eq:model-sp-sample-cons-i} are the copies of constraints \eqref{eq:model-d-1-cons-s}--\eqref{eq:model-d-1-cons-i} for scenarios $s=1,\ldots,S$.

\subsection{Risk Measures and Objective Functions}
\label{sec:risk-measure}
The risk measure $\rho$ in \eqref{eq:model-sp-sample-obj} in the stochastic optimization model can take different forms depending on the decision-maker's risk attitudes and uncertainty levels.
One of the most widely used measures is the expectation of a random variable, which measures the average performance, also known as the risk-neutral measure \citep{shapiro2014lectures}.
For any stochastic function $f(\xi)$ on the measure space $(\Xi, 2^\Xi, \mathcal{P})$, the expectation is defined as
    $\mathbb{E}_\xi[f(\xi)] = \int_{\Xi} f(\xi) d\mathcal{P}$.
With a sample set $\mathcal{S}$, the approximation expectation formulation for the stochastic function $f(\xi)$ is $\displaystyle \sum_{s=1}^S p_sf(\xi^s)$.

However, a particular realization of $\xi$ can be significantly different from its expectation.
In some quantum control applications, avoiding extremely poor performance of the given systems is important and necessary.
Here we consider a risk-averse measure, to control the risk in the solutions given by the stochastic optimization model \ref{eq:model-sp-sample}
\citep{sarykalin2008value,rockafellar2000optimization,rockafellar2002conditional,pflug2000some}.
We consider the Conditional Value-at-Risk (CVaR) \citep{rockafellar2000optimization} because it is a coherent risk measure with nice properties such as convexity.
For any stochastic function $f(\xi)$, the CVaR with risk level $\eta$ is defined as the expected value of $f(\xi)$ subject to the constraint that the value of $f(\xi)$ is no less than the lower $1-\eta$ percentile \citep{sarykalin2008value}.
\Cref{fig:cvar-illustration} illustrates CVaR for $\eta=0.05$, where the blue line at 2 represents the 95th percentile of $f(\xi)$ and the red dashed line represents the CVaR value as the average of all values of $f(\xi)$ that are larger than 2.

\begin{figure}[t]
    \centering
    \includegraphics[scale=0.55]{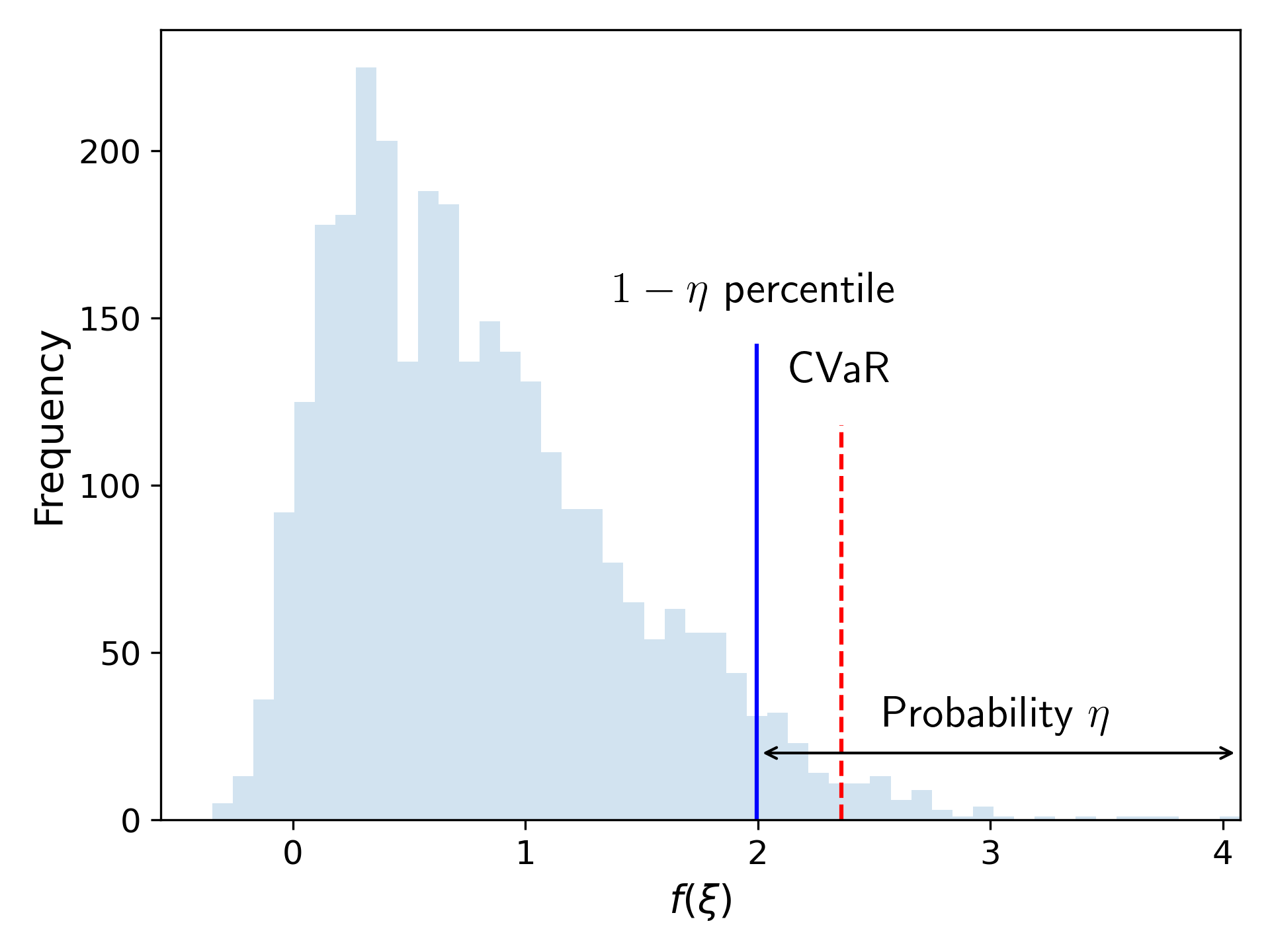}
    \caption{Illustration for CVaR function with risk level $\eta=0.05$. The histograms are the distribution of $f(\xi)$. The blue line represents the 95 percentile of $f(\xi)$. The red dashed line represents the CVaR value. }
    \label{fig:cvar-illustration}
\end{figure}

We also introduce the equivalent formulation for the CVaR function \citep{rockafellar2000optimization}:
\begin{align}
\label{eq:cvar-int}
   \mathrm{CVaR}_\eta \left(f(\xi)\right) = \inf\left\{\zeta + \frac{1}{\eta}\int_{\Xi} \max\{0, f(\xi) - \zeta \} d\mathcal{P} : \zeta\in \mathcal{R} \right\}.
\end{align}
The CVaR of $f(\xi)$ has the following sample-based approximation form:
\begin{align}
    \label{eq:cvar-sample}
    \min_{\zeta\in \mathcal{R}}\left( \zeta + \frac{1}{\eta} \sum_{s=1}^S p_s \max\{0, f(\xi^s)-\zeta\}\right).
\end{align}

One can consider a linear combination of expectation and CVaR function as the specific objective function in \ref{eq:model-sp-sample} to balance between risk-neutral and risk-averse attitudes. With a sample set $\mathcal{S}$, the risk measure~\eqref{eq:model-sp-sample-obj} can be formulated:
\begin{align*}
    \rho\left[F_X(X_T^1),\ldots,F_X(X_T^S)\right] = \alpha \sum_{s=1}^S p_s F_X(X_T^s) + (1-\alpha)\min_{\zeta\in \mathcal{R}}\left(\zeta + \frac{1}{\eta} \sum_{s=1}^S p_s \max\{0, F_X(X_T^s)-\zeta\}\right),
\end{align*}
where $\alpha\in [0,1]$ is a weight parameter and $\eta$ is a risk level parameter. When $\alpha=0$, the problem is equivalent to minimizing the CVaR function to obtain a risk-averse control. When $\alpha=1$, the goal is to optimize the expected performance of the control.

\section{Gradient-Based Algorithm}
\label{sec:algorithm}
We now present our algorithm for solving the stochastic optimization model, which consists of two parts, first solving a continuous relaxation and then rounding the continuous solution.
In \Cref{sec:continuous} we first convert the stochastic optimization model with sample approximation \ref{eq:model-sp-sample} to an unconstrained optimization model.
We then discuss the derivative for the objective function and introduce two gradient-based algorithms to solve the continuous relaxation of the model.
In \Cref{sec:rounding} we apply a sum-up rounding algorithm to obtain binary solutions with an optimality guarantee.

\subsection{Solution Methods for Continuous Relaxation}
\label{sec:continuous}
To start, we follow constraints~\eqref{eq:model-sp-sample-cons-h}--\eqref{eq:model-sp-sample-cons-i} and convert the final operator $X_T^s$ into an implicit function of control variables $u$:
\begin{align}
\label{eq:final-operator}
    X_T^s(u) = \prod_{k=1}^T \exp\left\{-i\left((1+\xi^s_{0k})H^{(0)} + \sum_{j=1}^N (1+\xi^s_{jk}) u_{jk} H^{(j)} \right)\Delta t\right\}X_\mathrm{init},\ s=1,\ldots,S.
\end{align}
We use $F^s(u)$ to denote the objective function of $u$ given uncertainty realization $\xi^s$ by substituting $X_T^s(u)$ into the objective function~\eqref{eq:model-d-1-obj}; that is, $F^s(u)=F_X(X_T^s(u))$.
We penalize the SOS1 property~\eqref{eq:model-d-1-cons-u-sum} by a squared $L_2$ penalty function in the form of
\begin{align}
  \label{eq:F_L_def}
    F_L(u)=\sum_{k=1}^T \left(\sum_{j=1}^N u_{jk} - 1\right)^2.
\end{align}
By relaxing the binary constraints~\eqref{eq:model-d-1-cons-u}, the stochastic optimization model \ref{eq:model-sp-sample} is converted to an unconstrained optimization problem over a bounded feasible region as
\begin{align}
\label{eq:uncons}
     \min_{u\in [0,1]^{N\times T}} \  \alpha \sum_{s=1}^S p_s F^s(u) + (1-\alpha)\min_{\zeta}\left(\zeta + \frac{1}{\eta} \sum_{s=1}^S p_s \max\{0, F^s(u)-\zeta\}\right) + \theta F_L(u),
\end{align}
where $\theta$ is the penalty weight parameter for the SOS1 property.
The $\max$ function in the second term leads to the objective possibly being
nondifferentiable with respect to the variables $u$ and $\zeta$; we therefore
discuss the closed-form expression and then the derivative of the objective function in the following theorems.
For simplicity, we denote the second term without weight $(1-\alpha)$ by $F_\mathrm{CVaR}(u,\zeta)$: 
\begin{align}
    F_\mathrm{CVaR}(u,\zeta) = \zeta + \frac{1}{\eta} \sum_{s=1}^S p_s \max\{0, F^s(u)-\zeta\}.
\end{align}
We derive a closed-form expression of $\min_\zeta F_\mathrm{CVaR}(u,\zeta)$ as follows.
\begin{theorem}
\label{thm:form-cvar}
For a given control variable $u$, define $s^*(u)$ as the scenario number with the largest original objective value $F^{s^*(u)}(u)$ such that
\begin{align}
\label{eq:def-s-star}
    \sum_{s=1}^S p_s \mathbb{1}_{\{F^s(u) > F^{s^*(u)}(u)\}} \geq \eta.
\end{align}
Then the closed-form expression of $\min_\zeta F_\mathrm{CVaR}(u,\zeta)$ at point $u$ is given by
    \begin{align}
\label{eq:form-great}
    F_{C} (u) & =
    F^{s^*(u)}(u) + \frac{1}{\eta} \sum_{s:F^s(u)> F^{s^*(u)}(u)} p_s (F^s(u) - F^{s^*(u)}(u)).
\end{align}
\end{theorem}

\begin{proof}
To prove the closed-form expression, it is equivalent to prove that given any feasible control variable $u$, $\zeta^*(u) = F^{s^*(u)}(u)$ is an optimal solution for the minimization problem $ \min_{\zeta} F_\mathrm{CVaR} (u,\zeta)$.
When $\zeta < \zeta^*(u) = F^{s^*(u)}(u)$, we have
\begin{align}
    F_\mathrm{CVaR}(u,\zeta) =& \zeta + \frac{1}{\eta} \sum_{s=1}^S p_s \max\{0, F^s(u)-\zeta\} = \zeta + \frac{1}{\eta}\sum_{s:\zeta< F^s(u)} p_s (F^s(u)-\zeta) \nonumber\\
    =&F_\mathrm{CVaR}(u, \zeta^*(u)) + \zeta - \zeta^*(u) + \frac{1}{\eta} \sum_{s:\zeta< F^s(u)\leq \zeta^*(u)}
    p_s(F^s(u)-\zeta) \nonumber\\
    &+ (\zeta^*(u)-\zeta) \frac{1}{\eta}\sum_{s:F^s(u) > \zeta^*(u)} p_s \nonumber\\
    \geq& F_\mathrm{CVaR}\left(u,\zeta^*(u)\right) + \frac{1}{\eta} \sum_{s:\zeta< F^s(u)\leq \zeta^*(u)}
    p_s(F^s(u)-\zeta) \geq F_\mathrm{CVaR}(u,\zeta^*(u)).
\end{align}
The equalities directly follow from the definition of $F_\mathrm{CVaR}(u,\zeta)$.
The first inequality holds because of the definition of $s^*(u)$ such that $\sum_{s:F^s(u)>\zeta^*(u)} p_s \geq \eta$.
The last inequality holds because all the terms in the summation have $F^s(u)>\zeta$. Similarly, we can show that when $\zeta > \zeta^*(u) = F^{s^*(u)} (u)$, we have
\begin{align}
    F_\mathrm{CVaR}(u,\zeta) = &\zeta + \frac{1}{\eta} \sum_{s=1}^S p_s \max\{0, F^s(u)-\zeta\} = \zeta + \frac{1}{\eta}\sum_{s:\zeta< F^s(u)} p_s (F^s(u)-\zeta) \nonumber\\
    = & F_\mathrm{CVaR}(u,\zeta^*(u)) + \zeta - \zeta^*(u) + \frac{1}{\eta} \sum_{s:\zeta^*(u)< F^s(u)\leq \zeta}
    p_s(\zeta - F^s(u)) \nonumber\\
    & + (\zeta^*(u)-\zeta)\frac{1}{\eta}\sum_{s:F^s(u)> \zeta^*(u)} p_s \nonumber\\
    \geq & F_\mathrm{CVaR}(u,\zeta^*(u)) + \frac{1}{\eta} \sum_{s:\zeta^*(u)< F^s(u)\leq \zeta}
    p_s(\zeta - F^s(u)) \geq F_\mathrm{CVaR}(u,\zeta^*(u)).
\end{align}
The only difference is that for the last inequality it holds because all the terms in the summation have $F^s(u)\leq \zeta$.
\end{proof}
\begin{remark}
    The smallest optimal solution of the minimization problem $\min_\zeta F_{\mathrm{CVaR}}(u,\zeta)$ is the value-of-risk (VaR) function value with risk level $\eta$ (see \citet{pang2004global}).
    Because scenarios $s$ with $F^s(u)-F^{s^*(u)}(u)=0$ contribute nothing to the summation, the closed-form expression~\eqref{eq:form-great} is equivalent to $F^{s^*(u)}(u) + \frac{1}{\eta} \sum_{s:F^s(u)\geq F^{s^*(u)}(u)} p_s (F^s(u) - F^{s^*(u)}(u))$.
\end{remark}
\begin{remark}
For a special case where we sample the scenario with equal probability, that is, $p_s=\frac{1}{S}$, an optimal solution $\zeta^*=F^{s^*(u)}(u)$ is the $\left\lceil \eta S\right\rceil$ largest original objective function value $F^s(u)$ among all the scenarios $s=1,\ldots,S$.
\end{remark}
Using the closed-form expression in~\eqref{eq:form-great},
we convert the original problem to minimizing an unconstrained continuous relaxation with uncertainty sample set $\mathcal{S}=\{\xi^1,\ldots,\xi^S\}$.
For simplicity, in the remaining discussion we define the summation of terms coming from the original objective function as
\begin{align}
    \label{eq:model-sp-sample-obj-convert}
    \tilde{F}(u) = \alpha \sum_{s=1}^S p_s F^s(u) + (1-\alpha)F_C(u),
\end{align}
where $F_C$ is defined in \eqref{eq:form-great}.

Taking the $L_2$ penalty function into consideration, the unconstrained continuous relaxation is
\begin{equation}
\makeatletter
\def\@currentlabel{SP-R($\mathcal{S}$)}
\makeatother
\label{eq:model-sp-sample-r}
    (\mathrm{SP-R}(\mathcal{S}))\quad \min_{u\in [0,1]^{N\times T}} \  F_R(u) = \alpha \sum_{s=1}^S p_s F^s(u) + (1-\alpha)F_C(u) + \theta F_L(u),
\end{equation}
which is a relaxation of our original stochastic optimization model \ref{eq:model-sp-sample} with relaxed binary constraints and penalized SOS1 property~\eqref{eq:model-sp-samplecons-u}.
The differentiability of the term $F_C(u)$ depends on the objective values of all the scenarios. For a given control variable point $\hat{u}$, we present the following theorem about the derivative.
\begin{theorem}
\label{thm:derivative-cvar}
For any given control variable point $\hat{u}$, if $F^s(\hat{u})\neq F^{s^*(\hat{u})} (\hat{u}),\ \forall s\neq s^*(\hat{u})$, then the closed form $F_C(u)$ is differentiable at point $\hat{u}$, with the derivative formulation as
\begin{align}
    \frac{\partial F_C(\hat{u})}{\partial \hat{u}} = \left(1 - \frac{1}{\eta} \sum_{s:F^s(\hat{u})>F^{s^* (\hat{u})}(\hat{u})} p_s\right) \frac{\partial F^{s^*(\hat{u})} (\hat{u})}{\partial \hat{u}} + \frac{1}{\eta} \sum_{s:F^s(\hat{u})>F^{s^* (\hat{u})}(\hat{u})} p_s \frac{\partial F^s(\hat{u})}{\partial \hat{u}}.
\end{align}
\end{theorem}

\begin{proof}

Because the functions $F_X(X_T)$ and $X_T^s(u)$ are differentiable for each scenario $s=1,\ldots,S$, the function $F^s(u)$ is differentiable and thus continuous.
Because the objective function $F^s(u)$ is continuous, given a control variable point $\hat{u}$, for any $\omega > 0$, there exists a distance $r_s(\hat{u}, \omega)$ for each scenario $s$ such that for any $\|u-\hat{u}\|\leq r_s(\hat{u}, \omega)$, we have $|F^s(u) - F^s(\hat{u})|< \omega$.
Choosing $\omega = \min_{s\neq s^*(\hat{u})} |F^s(\hat{u}) - F^{s^* (\hat{u})} (\hat{u})|$, from the assumption that $\forall s\neq s^*(\hat{u})$, $F^s(\hat{u})\neq F^{s^* (\hat{u})} (\hat{u})$, we have $\omega > 0$.
Define $r = \min_{s=1,\ldots,S} r_s(\hat{u}, \omega/2)$. Then, for any $u$ such that $|u - \hat{u}|\leq r$, $s^*(u)=s^*(\hat{u})$. We prove this claim by the following statements that
\begin{align}
    F^s(u)>F^s(\hat{u}) - \frac{\omega}{2} \geq F^{s^*(\hat{u})} (\hat{u}) + \omega - \frac{\omega}{2} = F^{s^*(\hat{u})}(\hat{u}) +\frac{\omega}{2} > F^{s^*(\hat{u})}(u),\ \forall s:F^s(\hat{u}) > F^{s^*(\hat{u})}(\hat{u})\\
    F^s(u)<F^s(\hat{u}) + \frac{\omega}{2} \leq F^{s^*(\hat{u})} (\hat{u}) - \omega + \frac{\omega}{2} = F^{s^*(\hat{u})}(\hat{u}) -\frac{\omega}{2} < F^{s^*(\hat{u})}(u),\ \forall s:F^s(\hat{u}) < F^{s^*(\hat{u})}(\hat{u}).
\end{align}
For both formulas, the first and last inequalities follow from the continuity of $F^s(u)$, and the other inequalities follow from the definition of $\omega$.
Now we show that $s^*(\hat{u})$ is still the scenario number with the largest original objective value such that $\displaystyle \sum_{s=1}^S p_s \mathbb{1}_{\{F^s(u) > F^{s^*(u)}(u)\}} \geq \eta$, which means that $s^*(u)=s^*(\hat{u})$. Furthermore, we show that $\{s:F^s(u)>F^{s^*(u)}(u)\} = \{s: F^s(\hat{u}) > F^{s^* (\hat{u})}(\hat{u})\}$.
Therefore, the derivative of $F_C(u)$ at point $\hat{u}$ is
\begin{align}
    \frac{\partial F_C(\hat{u})}{\partial \hat{u}} & = \lim_{u\rightarrow \hat{u}} \frac{F_C(u) - F_C(\hat{u})}{u - \hat{u}}\nonumber\\
    & = \lim_{u\rightarrow \hat{u}} \left(\frac{F^{s^*(\hat{u})}(u) - F^{s^*(\hat{u})}(\hat{u})}{u - \hat{u}} +
    \frac{1}{\eta} \sum_{s:F^s(u) > F^{s^*(\hat{u})}(u)} p_s \frac{F^s(u) - F^s(\hat{u})-F^{s^*(\hat{u})}(u) + F^{s^*(\hat{u})}(\hat{u})}{u - \hat{u}}\right)\nonumber\\
    & = \left(1 - \frac{1}{\eta} \sum_{s:F^s(\hat{u})>F^{s^* (\hat{u})}(\hat{u})} p_s\right) \frac{\partial F^{s^*(\hat{u})} (\hat{u})}{\partial \hat{u}} + \frac{1}{\eta} \sum_{s:F^s(\hat{u})>F^{s^* (\hat{u})}(\hat{u})} p_s \frac{\partial F^s(\hat{u})}{\partial \hat{u}}.
\end{align}
\end{proof}
When the conditions in~\Cref{thm:derivative-cvar} hold,
with the derivative of $F_C(u)$, we compute the derivative of the objective function $F_R(u)$ in~\ref{eq:model-sp-sample-r} by the chain rule as
\begin{align}
\label{eq:derivative}
    \frac{\partial F_R(u)}{\partial u_{jk}} = \alpha \sum_{s=1}^S p_s \frac{\partial F^s(u)}{\partial u_{jk}} + (1-\alpha) \frac{\partial F_C(u)}{\partial u_{jk}} + 2\theta \left( \sum_{j=1}^N u_{jk} - 1\right),\ j=1,\ldots,N,\ k=1,\ldots,T,
\end{align}
where the gradient of the original objective functions $F^s(u)$ for every scenario $s=1,\ldots,S$ depends on specific quantum problems and can be computed by the popular GRAPE algorithm~\citep{khaneja2005optimal}.
We apply two optimization methods, L-BFGS-B~\citep{byrd1995limited} and Adam~\citep{kingma2014adam}, with our derived gradient for $F_R(u)$ in~\eqref{eq:derivative} to solve the continuous relaxation of the stochastic optimization model.
In numerical studies we empirically show that L-BFGS-B is better for quantum problems aiming to minimize the energy of a quantum state, while Adam performs better on problems minimizing the infidelity compared with a target quantum operator.

\paragraph{L-BFGS-B algorithm} L-BFGS-B is a widely used quasi-Newton method for optimizing unconstrained models with deterministic objective functions. We first generate $S$ samples for the uncertain parameters $\xi$, then apply L-BFGS-B to solve \ref{eq:model-sp-sample-r}. Specifically, during each iteration of L-BFGS-B, we compute the derivative using~\eqref{eq:derivative} and the search direction, then conduct a line search to update control variables, following the details in~\citet{byrd1995limited,zhu1997algorithm}.

\paragraph{Adam method} Adam is a popular first-order gradient-based optimization method for optimizing unconstrained models with stochastic objective functions~\citep{kingma2014adam}.
We modify Adam to solve our problem with a bounded feasible region by adding a projection step.
The details of the algorithm are presented in \Cref{alg:adam} (where $\otimes$ represents elementwise multiplication between two vectors).
Specifically, during each iteration we first generate $S$ samples of the uncertain parameters $\xi$ to formulate the corresponding continuous relaxation \ref{eq:model-sp-sample-r} (see \Cref{algline:sample-model}). Then we compute the derivative by \Cref{eq:derivative} (see \Cref{algline:derivative}), update the control variables, and project the updated variables to the feasible region $[0,1]^{N\times T}$ (see \Crefrange{algline:update-var-start}{algline:update-var-end}).
\begin{algorithm}[!ht] \caption{Adam for solving the continuous relaxation of the stochastic optimization model. \label{alg:adam} }
    \DontPrintSemicolon
    \SetNoFillComment
    \KwInput{Initial control values $u^{(0)}$. Maximum iteration number $K$.\;}
    \KwInput{Step size $\gamma_1> \gamma_2>0$. Objective value threshold to change step size $\bar{F}$.\;}
    \KwInput{Decay rates for the moment estimates $\beta_1,\ \beta_2\in [0, 1)$. Constant for numerical computation $\epsilon$.\;
    }
    Initialize the first and second moment vectors $m^{(0)},\ v^{(0)}$ as $N\times T$-dimensional zero vectors.\; Initialize step size $\hat{\gamma}=\gamma_1$.\;
    \For{Iteration $i=1,\ldots,K$}
    {
        Generate a sample set $\mathcal{S}=\{\xi^1,\ldots,\xi^S\}$ with probability $p_s$ for $\xi^s,\ s=1,\ldots,S$ and formulate the unconstrained sample approximation \ref{eq:model-sp-sample-r}.\;
        \label{algline:sample-model}
        Compute the corresponding derivative $\displaystyle g^{(i)}=\frac{\partial F_R(u)}{\partial u^{(i-1)}}$ by Equation~\eqref{eq:derivative}.\;
         \label{algline:derivative}
        Compute the first moment vector $\displaystyle m^{(i)} = \left(\beta_1 m^{(i-1)} + (1-\beta_1) g^{(i)}\right)$.\;
        \label{algline:update-var-start}
        Compute the second moment vector $\displaystyle v^{(i)} = \left(\beta_2 v^{(i-1)} + (1-\beta_2) g^{(i)} \otimes g^{(i)}\right)$.\;
        Compute the bias-corrected moment vectors $\displaystyle \hat{m}^{(i)} = m^{(i)}/ (1 - \beta_1^i),\ \hat{v}^{(i)} = v^{(i)} / (1 - \beta_2^i)$.\;
        Update control variables and project to the feasible bounded region $[0,1]^{N\times T}$ as $u^{(i)} = \Pi_{[0,1]^{N\times T}} \left(u^{(i-1)}-\hat{\gamma} \hat{m}^{(i)}/(\sqrt{\hat{v}^{(i)}} + \epsilon)\right)$.\;
        \label{algline:update-var-end}
        \If{$F_R(u^{(i)})<\bar{F}$}
        {$\hat{\gamma} = \gamma_2$.}
    }
    \KwOutput{Continuous control solutions $u^{(K)}$.}
\end{algorithm}

\subsection{Sum-Up Rounding Technique}
\label{sec:rounding}
With continuous solutions $u^\mathrm{con}\in [0,1]^{N\times T}$, we apply the sum-up rounding (SUR) technique to obtain binary solutions $u^\mathrm{bin}$. The SUR technique is proposed by the work of \citet{sager2012integer} and is widely used in integer control optimization problems.
To the best of our knowledge, most work using SUR rounds either a continuous-time control function~\citep{sager2011combinatorial,sager2021mixed} or controls of the continuous relaxation with the same time discretization~\citep{leyffer2021convergence,manns2020multidimensional,Fei2023binarycontrolpulse,fei2023switching}.

In our problem, the time of solving the continuous relaxation is the major part of the overall computational time and significantly increases when the number of time steps $T$ and scenarios $S$ is high.
Therefore, we solve the continuous relaxation using fewer time steps $T$ and round the solutions using more time steps $T_R$ to achieve a balance between computational time and the difference between continuous and binary solutions.
For simplicity, we assume here that $T_R=C_\mathrm{SUR}T$, where $C_\mathrm{SUR} > 1$ is a predetermined integer constant.
We present the rounding algorithm procedure in \Cref{alg:sur}.

\begin{algorithm}[!ht] \caption{Sum-up rounding for continuous and binary solutions with different time steps. \label{alg:sur} }
    \DontPrintSemicolon
    \SetNoFillComment
    \KwInput{Time steps of continuous solution $T$. The multiplier factor between time steps of continuous and binary solutions $C_\mathrm{SUR}$. Continuous control $u^\mathrm{con}\in [0,1]^{N\times T}$. }
    \For{$k=1,\ldots,C_\mathrm{SUR} T$}
    {\For{$j=1,\ldots,N$}
    {Compute cumulative deviation
    $\displaystyle \delta_{jk} = \sum_{\tau=1}^k u^\mathrm{con}_{j\lfloor \tau /C_\mathrm{SUR} \rfloor} \frac{\Delta t} {C_\mathrm{SUR}} -\sum_{\tau=1}^{k-1} u^\mathrm{bin}_{j\tau } \frac{\Delta t} {C_\mathrm{SUR}}$.}
    Let $j^* = \argmax_{j=1,\ldots,N} \delta_{jk}$, breaking ties by by choosing the smallest index.\;
    \label{algline:sur-select}
       Update binary control $u^\mathrm{bin}_{j^*k}=1$ and $u^\mathrm{bin}_{jk}=0,\ \forall j\neq j^*$.
       \label{algline:sur-update}
    }
    \KwOutput{Binary control $u^\mathrm{bin}\in \{0,1\}^{N\times C_\mathrm{SUR}T}$.}
\end{algorithm}
\begin{remark}
    \Cref{alg:sur} can be extended to a more general case when the SOS1 property is not required for controls by only changing~\Crefrange{algline:sur-select}{algline:sur-update}.
    We set $u^\mathrm{bin}_{jk}=1,\ \forall j\in J^*$ where $J^*=\{\delta_{jk}\geq 0.5 \Delta t / C_\mathrm{SUR}\}$ and other controller values as $0$.
\end{remark}
In the rest of this section we discuss how the difference between continuous and binary controls varies with time steps $T$.
In the remaining discussion we use $u$ to represent all the discretized controls and $u(t)$ to represent all the control functions on a continuous time horizon (i.e., $T= \infty$).
We first propose two assumptions for the original problem, which are satisfied in most quantum control problems.
\begin{assumption}
\label{assump:1-obj}
    We assume that the original objective function for each quantum system $F_X$ is continuous, non-negative, and upper-bounded.
\end{assumption}
\begin{assumption}
\label{assump:2-feasible}
     We assume that the stochastic optimization model \ref{eq:model-sp-sample} is feasible.
\end{assumption}
We define piecewise constant control functions $u^\mathrm{con}(t)$ and $u^\mathrm{bin}(t)$ as equivalent formulations to discretized controls $u^\mathrm{con}$ and $u^\mathrm{bin}$:
\begin{subequations}
    \begin{align}
        & u^\mathrm{con}_j(t) = u^\mathrm{con}_{jk},\ \forall t\in [(k-1)\Delta t, k\Delta t),\ j=1,\ldots,N,\ k=1,\ldots,T.\\
        & u^\mathrm{bin}_j(t) = u^\mathrm{bin}_{jk},\ \forall t\in [(k-1)\frac{\Delta t}{C_\mathrm{SUR}}, k\frac{\Delta t}{C_\mathrm{SUR}}),\ j=1,\ldots,N,\ k=1,\ldots,C_\mathrm{SUR}T.
    \end{align}
\end{subequations}
In the following theorem we discuss the cumulative difference between continuous and binary control functions:
\begin{theorem}
\label{thm:cum-diff-bound}
    With \Crefrange{assump:1-obj}{assump:2-feasible}, let $\FUB$ be the upper bound of $F_X$. Then the cumulative difference between continuous and binary controls at any time $t$ satisfies
    \begin{align}
    \label{eq:cum-diff}
         \left\|\int_{0}^t \left( u^\mathrm{con}(\tau) - u^\mathrm{bin}(\tau)\right)d\tau \right\|_\infty & \leq \frac{\left(N-1\right)}{C_\mathrm{SUR}}\Delta t + \frac{2N - 1}{N} \sqrt{t_fF_L(u^\mathrm{con})\Delta t},\ \forall t\in [0,t_f],
    \end{align}
    where $F_L$ is defined in \eqref{eq:F_L_def}.
    Furthermore, we have the following convergence results for the objective values defined in~\eqref{eq:model-sp-sample-obj-convert} of continuous and binary solutions:
    \begin{align}
        \lim_{\Delta t \rightarrow 0} \tilde{F}(u^\mathrm{con}) =  \lim_{\Delta t\rightarrow 0} \tilde{F}(u^\mathrm{bin}).
    \end{align}
\end{theorem}
\begin{proof}
    For any time interval length $\Delta t$, let $\hat{k}$ be the index of the time step in the SUR algorithm that $t$ falls in. Then the integral can be written as follows based on the definition of piecewise constant functions:
    \begin{align}
        \int_{0}^t \left( u^\mathrm{con}(\tau) - u^\mathrm{bin}(\tau)\right)d\tau = \sum_{\tau=0}^{\hat{k}} \left(u^\mathrm{con}_{j\lfloor \tau /C_\mathrm{SUR} \rfloor} - u^\mathrm{bin}_{j\tau}\right)\frac{\Delta t}{C_{\mathrm{SUR}}}.
    \end{align}
    The remaining proof of the upper bound in~\eqref{eq:cum-diff} directly follows the proof of Theorem 4 and Corollary 1 in~\citet{Fei2023binarycontrolpulse} with details being omitted here.
    For any control value $u$, we have
    \begin{align}
        F_C(u) = \min_\zeta F_{\mathrm{CVaR}} (u, \zeta) \leq F_{\mathrm{CVaR}} (u, \FUB) = \FUB.
    \end{align}
    Hence, the original risk measure function $\tilde{F}(u)$ defined in~\eqref{eq:model-sp-sample-obj-convert} is upper bounded by $\FUB$.
    Because $u^\mathrm{con}$ is the optimal solution of penalized continuous relaxation, for a feasible solution $u^f$ of model \ref{eq:model-sp-sample} we have
    \begin{align}
        \tilde{F}(u^\mathrm{con}) + F_L(u^\mathrm{con}) \leq \tilde{F}(u^f) + F_L(u^f) = \tilde{F}(u^f)\leq \FUB,
    \end{align}
    where the last equality follows by the fact that $u^f$ is a feasible solution of the model with the SOS1 property, so $F_L(u^f)=0$.
    Therefore, we have $F_L(u^\mathrm{con})\leq \FUB$, and the convergence of objective values directly follows Corollary 8 in~\citet{sager2012integer}.
\end{proof}

\Cref{thm:cum-diff-bound} proves that $F_L(u^\mathrm{con})\leq \FUB$ and the cumulative difference is upper bounded by $O(\sqrt{\Delta t})$.
In the following propositions we show that with additional assumptions the upper bound can be tightened.
We first introduce the infinite-dimension formulation (i.e., $T= \infty$) for the original objective function of a single scenario $s$:
\begin{align}
\label{eq:obj-single-infinite}
    F^s(u(t)) = F_X(X^s(t_f; u)),\ s=1,\ldots,S.
\end{align}
Note that we use $F^s(u)$ for the objective value of the discretized control and $F^s(u(t))$ for the objective value of the continuous-time control.

The operator $X^s(t_f;u)$ is the value of $X^s(t;u)$ at time $t_f$, and $X^s(t;u)$ is the solution for the following differential equation with given control functions $u(t)$:
\begin{align}
    \frac{d}{dt} X^s(t) = -i\left((1+\xi_0^s(t)) H^{(0)} + \sum_{j=1}^N (1+\xi_j^s(t))u_j(t) H^{(j)}\right)X^s(t),
\end{align}
where $\xi_0^s(t),\ldots,\xi_N^s(t)$ are time-dependent uncertain parameters.
We use the infinite-dimension objective function $F^s(u(t))$ in \eqref{eq:obj-single-infinite} to replace $F^s(u)$ in the stochastic objective function $\tilde{F}(u)$ defined in~\eqref{eq:form-great} and~\eqref{eq:model-sp-sample-obj-convert}.
We define $s^*(u(t))$ as the scenario number with the largest original objective value $F^{s^*(u(t))}$ such that $\displaystyle \sum_{s=1}^S p_s\mathbb{1}_{F^s(u(t))>F^{s^*(u(t))}(u(t))}\geq \eta$ (see the discretized version in \Cref{thm:form-cvar}). For simplicity, we use $s^*$ to replace $s^*(u(t))$ in the following formulation of the infinite-dimension stochastic objective function $\tilde{F}(u(t))$:
\begin{align}
    \label{eq:model-sp-sample-obj-conver-infinite}
    \tilde{F}(u(t)) & = \alpha \sum_{s=1}^S p_s F^s(u(t))\nonumber\\
    & + (1-\alpha) \left[F^{s^*}(u(t)) + \frac{1}{\eta} \sum_{s:F^s(u(t)) > F^{s^*} (u(t))} p_s\left(F^s(u(t)) - F^{s^*}(u(t))\right)\right],
\end{align}
With the definition of infinite-dimension stochastic objective function $\tilde{F}(u(t))$ in~\eqref{eq:model-sp-sample-obj-convert}, we define the infinite-dimension formulation with the SOS1 property for the stochastic optimization model \ref{eq:model-sp-sample} as
\begin{subequations}
        \makeatletter
\def\@currentlabel{SP-C}
\makeatother
    \label{eq:model-c-sp}
    \begin{align}
    \label{eq:model-c-sp-obj}
        (\mathrm{SP-C}) \quad
    \min_{u} \quad & \tilde{F}(u(t))\\
    \label{eq:model-c-sp-cons-sos1}
    \mathrm{s.t.}\quad   & \sum_{j=1}^N u_j(t)=1,\  a.e.\ t\in [0,t_f] \\
    \label{eq:model-c-sp-cons-binary}
    & u_{j}(t)\in \left\{0,1\right\},\ j=1,\ldots,N,\ a.e.\ t\in [0, t_f].
    \end{align}
\end{subequations}
The objective function~\eqref{eq:model-c-sp-obj} is the stochastic objective function defined in~\eqref{eq:model-sp-sample-obj-conver-infinite}.
Constraint~\eqref{eq:model-c-sp-cons-sos1} enforces that the control function holds the SOS1 property for $t\in [0,t_f]$.
Constraint~\eqref{eq:model-c-sp-cons-binary} indicates that the control function value is binary for $t\in [0,t_f]$~\citep{Wikipedia2023Almost}.
Every feasible solution of the discretized model \ref{eq:model-sp-sample} can be considered as a piecewise constant control function and thus is a feasible solution for the infinite-dimension formulation~\ref{eq:model-c-sp}.
For this model we impose the following assumption and derive an $O(\Delta t)$ upper bound for the cumulative difference based on it.
\begin{assumption}
    \label{assump:3-lb}
    We assume that there exists an optimal solution for the continuous relaxation of the infinite-dimension model with the SOS1 property~\ref{eq:model-c-sp}, represented by $u^{*,\mathrm{SOS1}}(t)$ such that the original objective value $\tilde{F}(u^{*, \mathrm{SOS1}}(t))=0$.
\end{assumption}
\begin{proposition}
\label{prop:cum-diff-ext1}
Recall that $\theta$ is the weight parameter of the SOS1 $L_2$ penalty function, with \Crefrange{assump:1-obj}{assump:3-lb}. Then we have the following bound for the cumulative difference:
    \begin{align}
    \label{eq:cum-diff-ext1}
        \left\|\int_{0}^t \left( u^\mathrm{con}(\tau) - u^\mathrm{bin}(\tau)\right)d\tau \right\|_\infty & \leq
        \left(\frac{\left(N-1\right)}{C_\mathrm{SUR}} + \frac{2N-1}{N\sqrt{\theta}} C_{\mathrm{diff}}\right) \Delta t ,\ \forall t\in [0,t_f],
    \end{align}
    where $C_{\mathrm{diff}}$ is a constant determined by control Hamiltonians and evolution time $t_f$.
\end{proposition}
\begin{proof}
    We only need to prove that there exists a constant $C_{\mathrm{diff}}$ such that
    $\sqrt{t_f F_L(u^\mathrm{con})}\leq C_{\mathrm{diff}} \sqrt{\Delta t}/\sqrt{\theta}$.
    We define $u^{\mathrm{con,SOS1}}$ as the optimal solution for the continuous relaxation of the discretized model with the SOS1 property \ref{eq:model-sp-sample}. From the optimality of $u^\mathrm{con}$ we have
    \begin{align}
    \Tilde{F}(u^\mathrm{con}) + \theta F_L(u^\mathrm{con}) \leq \Tilde{F}(u^\mathrm{con,SOS1}) + \theta F_L(u^\mathrm{con,SOS1}) = \Tilde{F}(u^\mathrm{con,SOS1}),
\end{align}
where the inequality follows the fact that $u^\mathrm{con,SOS1}$ holds the SOS1 property, so the penalty term $F_L(u^\mathrm{con,SOS1})=0$. Combining with the \Cref{assump:1-obj} that $\tilde{F}(u^{\mathrm{con}})\geq 0$, we have
    \begin{align}
    \label{eq:ub-penalty}
        F_L(u^\mathrm{con})\leq \frac{1}{\theta} \tilde{F}(u^\mathrm{con,SOS1}).
    \end{align}
    We then consider the difference in objective values between the optimal infinite-dimension relaxation solution $u^{*,\mathrm{SOS1}}(t)$ (defined in \Cref{assump:3-lb}) and the optimal discretized relaxation solution $u^\mathrm{con,SOS1}$.
    We first construct a piecewise constant control function $u^{d, \mathrm{SOS1}}(t)$ satisfying the inequality as
    \begin{align}
    \label{eq:def-ud}
        u^{d, \mathrm{SOS1}}(t) = \frac{1}{\Delta t} \int_{(k-1)\Delta t}^{k\Delta t} u^{*, \mathrm{SOS1}}(\tau) d\tau,\ \forall t\in [(k-1)\Delta t, k\Delta t),\ k=1,\ldots,T.
    \end{align}
    It is obvious that during each time interval we have
    \begin{align}
        \int_{(k-1)\Delta t}^{k\Delta t} u^{d, \mathrm{SOS1}}(\tau) d\tau = \int_{(k-1)\Delta t}^{k\Delta t} u^{*, \mathrm{SOS1}}(\tau) d\tau,\ k=1,\ldots,T.
    \end{align}
    For any time $t\in [0,t_f]$, let $\hat{k}$ be the index of time interval that $t$ falls in. Then we have
    \begin{align}
        \left\|\int_0^t u^{d, \mathrm{SOS1}}(\tau) d\tau - \int_0^tu^{*, \mathrm{SOS1}}(\tau) d\tau\right\|_\infty
         \leq \left\|\int_{(\hat{k}-1)\Delta t}^t u^{d, \mathrm{SOS1}}(\tau) d\tau - \int_{(\hat{k}-1)\Delta t}^t u^{*, \mathrm{SOS1}}(\tau) d\tau\right\|_\infty .
    \end{align}
    In the time subinterval $[(\hat{k}-1)\Delta t, \hat{k} \Delta t]$, the two integrals hold:
    \begin{subequations}
            \begin{align}
        \int_{(\hat{k}-1)\Delta t}^t u^{d, \mathrm{SOS1}}(\tau) d\tau \leq \max_{\tau \in [(\hat{k}-1)\Delta t, \hat{k}\Delta t]} u^{d, \mathrm{SOS1}}(\tau) \Delta t\\
        \int_{(\hat{k}-1)\Delta t}^t u^{*, \mathrm{SOS1}}(\tau) d\tau \geq \min_{\tau \in [(\hat{k}-1)\Delta t, \hat{k}\Delta t]} u^{*, \mathrm{SOS1}}(\tau) \Delta t .
    \end{align}
    \end{subequations}
    From the definition of $u^{d, \mathrm{SOS1}}$ in~\eqref{eq:def-ud}, we know that
    \begin{align}
        \displaystyle \max_{\tau \in [(\hat{k}-1)\Delta t, \hat{k}\Delta t]} u^{d, \mathrm{SOS1}}(\tau)\leq \max_{\tau \in [(\hat{k}-1)\Delta t, \hat{k}\Delta t]} u^{*, \mathrm{SOS1}}(\tau).
    \end{align}
    Therefore, for any $t \in [0, t_f]$,
    \begin{align}
    \label{eq:int-bound-d-opt}
        &\left\|\int_0^t u^{d, \mathrm{SOS1}}(\tau) d\tau - \int_0^tu^{*, \mathrm{SOS1}}(\tau) d\tau\right\|_\infty \nonumber\\\leq & \left\|\max_{\tau \in [(\hat{k}-1)\Delta t, \hat{k}\Delta t]} u^{*, \mathrm{SOS1}}(\tau) - \min_{\tau \in [(\hat{k}-1)\Delta t, \hat{k}\Delta t]} u^{*, \mathrm{SOS1}}(\tau)\right\|_\infty \Delta t.
    \end{align}
     We notice that the values of control functions $u^{*, \mathrm{SOS1}}(t)$ and $u^{d, \mathrm{SOS1}}(t)$ are both bounded by $[0,1]$. Hence, the difference of integral is upper bounded by $\Delta t$.
     From Theorem 2 in~\citet{sager2021mixed} we have
    \begin{align}
        \left\|X(t_f; u^{*, \mathrm{SOS1}}) - X(t_f; u^{d, \mathrm{SOS1}})\right\| \leq C'\Delta t,
    \end{align}
    where $C'$ is a constant determined by control Hamiltonians and evolution time.
    Combining with the continuity of the objective function $\tilde{F}(u(t))$, we have
    \begin{align}
        \tilde{F}(u^{d, \mathrm{SOS1}}(t)) - \tilde{F}(u^{*, \mathrm{SOS1}}(t))\leq C''\Delta t,
    \end{align}
    where $C''$ is a constant determined by objective function $\tilde{F}(u)$.
    From the definition of the piecewise constant control function $u^\mathrm{d,SOS1}(t)$, we can construct an equivalent discretized solution $u^\mathrm{d,SOS1}_k = u^\mathrm{d,SOS1}(t)$, where $k$ is the index of time interval that $t$ falls in.
    Because $u^{\mathrm{con}, \mathrm{SOS1}}$ is the optimal solution of the discretized formulation, it holds that
    \begin{align}
      \tilde{F}(u^{\mathrm{con}, \mathrm{SOS1}}) - \tilde{F}(u^{*, \mathrm{SOS1}}(t)) &\leq \tilde{F}(u^{d, \mathrm{SOS1}}) - \tilde{F}(u^{*, \mathrm{SOS1}}(t))\nonumber\\
        & = \tilde{F}(u^{d, \mathrm{SOS1}}(t)) - \tilde{F}(u^{*, \mathrm{SOS1}}(t))\leq C''\Delta t.
    \end{align}
    With \Cref{assump:3-lb} that $\tilde{F}(u^{*, \mathrm{SOS1}}(t))=0$ and the upper bound for $F_L(u^\mathrm{con})$ in~\eqref{eq:ub-penalty}, we prove that
    \begin{align}
        \sqrt{t_f F_L(u^\mathrm{con})} \leq C_\mathrm{diff} \sqrt{\Delta t}/\sqrt{\theta},
    \end{align}
    where $C_\mathrm{diff} = \sqrt{t_fC''}$.
\end{proof}

Furthermore, we prove in \Cref{prop:cum-diff-ext2}, with an additional \Cref{assump:4-continuous} for the infinite dimension model~\ref{eq:model-c-sp} that the second term of the cumulative difference is upper bounded by $o(\Delta t)\Delta t$, where $\lim_{\Delta t\rightarrow 0}o(\Delta t)=0$.
\begin{assumption}
    \label{assump:4-continuous}
    We assume that there exists a constant time interval length $\Delta t_0$ such that for any time discretization with $\Delta t\leq \Delta t_0$, the optimal solution $u^{*, \mathrm{SOS1}}(t)$ for the continuous relaxation of infinite dimension model with the SOS1 property~\ref{eq:model-c-sp} is continuous in each time subinterval.
\end{assumption}
\begin{proposition}
    \label{prop:cum-diff-ext2}
    With \Crefrange{assump:1-obj}{assump:4-continuous}, for any $\Delta t\leq \Delta t_0$ we have the following bound for the cumulative difference:
    \begin{align}
    \label{eq:cum-diff-ext2}
        \left\|\int_{0}^t \left( u^\mathrm{con}(\tau) - u^\mathrm{bin}(\tau)\right)d\tau \right\|_\infty & \leq
        \frac{\left(N-1\right)}{C_\mathrm{SUR}} \Delta t+ \frac{2N-1}{N\sqrt{\theta}} o(\Delta t) \Delta t ,\ \forall t\in [0,t_f],
    \end{align}
    where $\lim_{\Delta t\rightarrow 0} o(\Delta t)=0$.
\end{proposition}
\begin{proof}
    We consider the upper bound for the cumulative difference between $u^{*,\mathrm{SOS1}}(t)$ and $u^{d,\mathrm{SOS1}}(t)$ in~\eqref{eq:int-bound-d-opt}. Because we assume that the optimal solution $u^{*,\mathrm{SOS1}}(t)$ is continuous in each subinterval in \Cref{assump:4-continuous}, the right-hand side is upper bounded by $o(\Delta t)\Delta t$. The other parts of the proof are the same as the proof of \Cref{prop:cum-diff-ext1}.
\end{proof}

The upper bound in \Cref{prop:cum-diff-ext2} indicates that the first term dominates the second term, and therefore increasing the multiplier factor for time steps in the SUR algorithm ($C_\mathrm{SUR}$) significantly reduces the cumulative difference between binary and continuous controls if \Crefrange{assump:1-obj}{assump:4-continuous} hold.

For a fixed number of time steps $T$, the optimal value of the continuous relaxation ($\tilde{F}(u^\mathrm{con})$) provides a lower bound for the binary model with the same time steps but not necessarily for the binary model with time steps $C_\mathrm{SUR}T$.
We provide a counterexample in the following remark.
\begin{remark}
We provide an example showing that the objective value of the continuous relaxation optimal solution $\tilde{F}(u^\mathrm{con})$ is larger than the objective value of the binary solution $\tilde{F}(u^\mathrm{bin})$ obtained by SUR with rounding time steps $C_\mathrm{SUR}T$.
    We consider a quantum control problem with zero noises as follows. The objective function is defined as
    \begin{align}
        1 - \frac{1}{4}\left|\operatorname{tr}\left\{X_\mathrm{targ}^\dagger X_T\right\}\right|,
    \end{align}
    where $X_\mathrm{targ}$ is
    \begin{align}
        X_\mathrm{targ} = \begin{pmatrix}
        1 & 0 & 0 & 0\\
        0 & 1 & 0 & 0\\
        0 & 0 & 0 & 1\\
        0 & 0 & 1 & 0
    \end{pmatrix}.
    \end{align}
    We set intrinsic and control Hamiltonians as
    \begin{align}
        H^{(0)} = \begin{pmatrix}
        1 & 0 & 0 & 0\\
        0 & -1 & 2 & 0\\
        0 & 2 & -1 & 1\\
        0 & 0 & 1 & 1
    \end{pmatrix},\ H^{(1)} = \begin{pmatrix}
        0 & 0 & 1 & 0\\
        0 & 0 & 0 & 1\\
        1 & 0 & 0 & 0\\
        0 & 1 & 0 & 0
    \end{pmatrix},\ H^{(2)} = \begin{pmatrix}
        0 & 0 & -i & 0\\
        0 & 0 & 0 & -i\\
        i & 0 & 0 & 0\\
        0 & i & 0 & 0
    \end{pmatrix}.
    \end{align}
    The initial operator $X_\mathrm{init}$ is a 4-dimensional identity matrix. The evolution time $t_f=8$.
    The number of time steps for continuous relaxation $T=1$ and the number of time steps for SUR $T_R=C_\mathrm{SUR}T=100$.
    We consider only one scenario $S=1$ and uncertainty $\xi_{jk}=0$ for all $j=0,1,2$ and $k=1$. The time-dependent Hamiltonians $H_k^s$ are computed by~\eqref{eq:model-sp-sample-cons-h}, and the time-dependent operators $X_k^s$ are computed by~\eqref{eq:model-sp-sample-cons-s}.
    By solving the model, we show that $\tilde{F}(u^\mathrm{bin})=0.673 < 0.678 = \tilde{F}(u^\mathrm{con})$.
\end{remark}

\section{Numerical Studies}
\label{sec:numeric}
We now apply the algorithms discussed in \Cref{sec:algorithm} to solve two quantum control problems with uncertain Hamiltonians: an energy minimization problem and a circuit compilation problem. In \Cref{sec:noise} we introduce our simulation setup for quantum systems with uncertain Hamiltonians.
In \Cref{sec:res-energy} we introduce the settings of the energy minimization problem and present the numerical results. In \Cref{sec:res-circuit} we introduce the circuit compilation problem and describe the numerical results.
All numerical simulations were conducted on a macOS computer with 8 cores, 16 GB RAM, and a 3.20 GHz processor.
The implementation was in Python with version 3.8. Our full code and results are available on our GitHub repository~\citep{coderobust2023}.

\subsection{Uncertainty Design}
\label{sec:noise}
Our proposed stochastic optimization model focues on quantum systems where the control Hamiltonians are inexact.
In realistic experiments, the control uncertainty of Hamiltonians varies among each simulation and each time step and has different distribution parameters for different Hamiltonians.
The variances of the uncertain parameters are larger across simulations and smaller within each time step during a single simulation.
Specifically, we assume that for each controller $j$ and each time step $k$, the uncertain parameter $\xi_{jk}$ follows a normal distribution $\mathcal{N}(\mu, \sigma_j^\mathrm{time})$, where $\mu$ is a random variable following a normal distribution $\mathcal{N}(0, \sigma_j^\mathrm{offset})$, with constant variance $\sigma_j^\mathrm{time}, \sigma_j^\mathrm{offset}$ determined by the Hamiltonians.

For each scenario $s=1,\ldots,S$, we generate the corresponding samples as follows.
We first sample a parameter $\mu^s_j$ from a normal distribution with mean value 0 and variance $\sigma^\mathrm{offset}_j$ for each Hamiltonian $H^{(j)},\ j=0,\ldots,N$, representing the mean value of uncertainties for Hamiltonian $H^{(j)}$ across all time steps for scenario $s$, defined as an offset.
We then sample $\xi^s_{jk}$ for each Hamiltonian $H^{(j)}$, for all $j=0,\ldots,N$ and time step $k=1,\ldots,T$ from a normal distribution with mean value $\mu^s_{j}$ and variance $\sigma^\mathrm{time}_j$.

We show the values of $10$ sampled scenarios of $\xi_{0k},\ \forall k=1,\ldots,T$ in \Cref{fig:uncertainty}. The different intercepts of lines reflect the variances of each simulation described by $\sigma_0^\mathrm{offset}$, and the fluctuation of each line indicates the variances among time steps for each simulation, described by $\sigma_0^\mathrm{time}$.

\begin{figure}[t]
    \centering
    \includegraphics[width=0.7\textwidth]{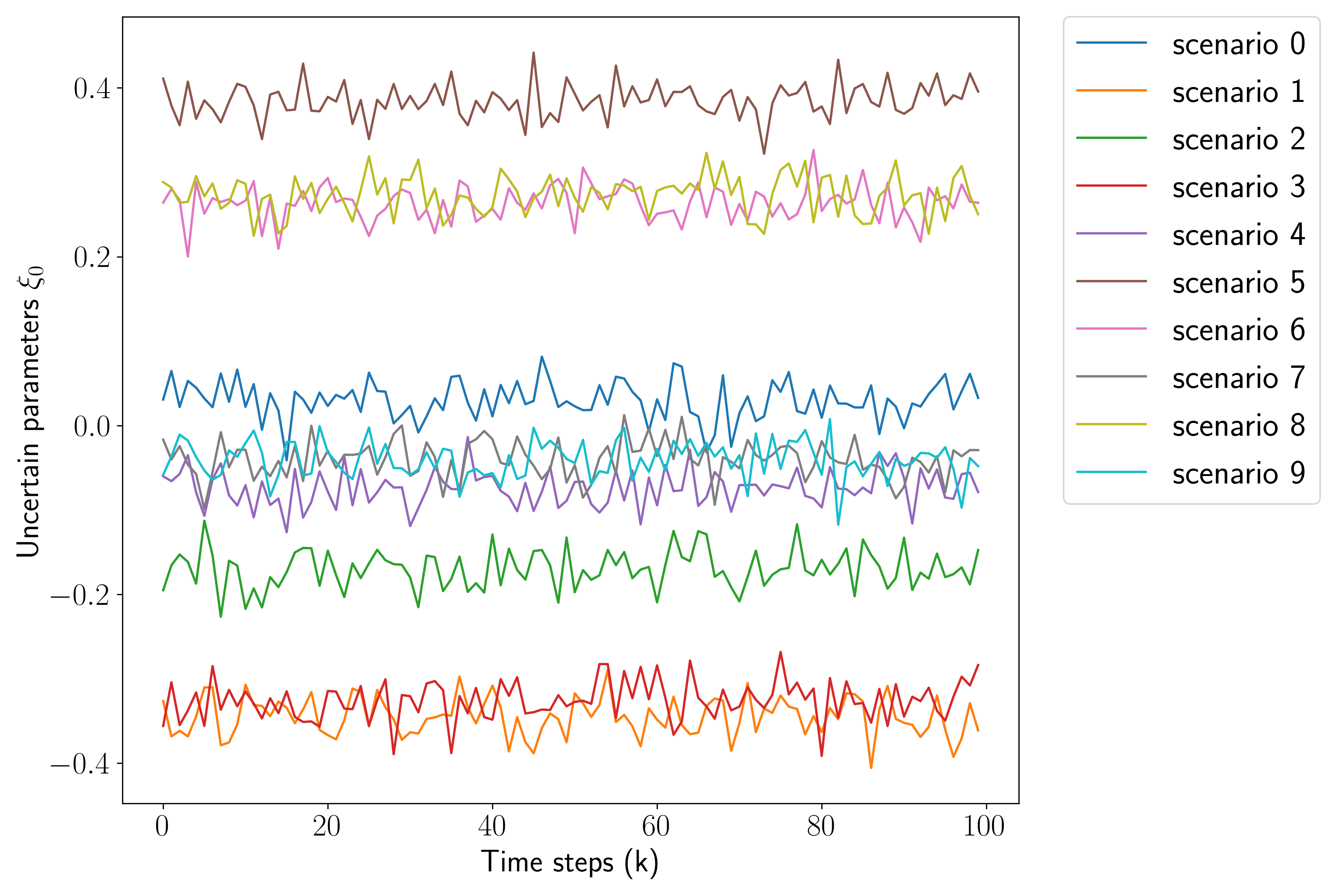}
    \caption{Sampled values of $\xi_{0}$ with $10$ scenarios. The x-axis is time step $k=1,\ldots,T$, and the y-axis is the value of $\xi_{0}$. The lines represent values of corresponding samples for each scenario $s=0,\ldots,9$.}
    \label{fig:uncertainty}
\end{figure}

\subsection{Energy Minimization Problem}
\label{sec:res-energy}
We apply L-BFGS-B to solve the stochastic optimization model of an energy minimization quantum control problem.
Consider a spin system with $q$ qubits, no intrinsic Hamiltonian, and two control Hamiltonians $H^{(1)},\ H^{(2)}$.
We define the initial state $|\psi_0\rangle$ as the ground state of $H^{(1)}$, which is the eigenvector of $H^{(1)}$ with minimum eigenvalue.
The goal is to minimize the energy corresponding to $H^{(2)}$ of the final state $X_T|\psi_0\rangle$.
Denote the theoretical minimum value of obtained energy as $E_{\mathrm{min}}$, which is the minimum eigenvalue of $H^{(2)}$.
The specific deterministic formulation is given by
\begin{subequations}
\label{eq:model-energy-c}
\begin{align}
\label{eq:model-energy-c-obj}
    \min\quad  \displaystyle & 1 - \left \langle \psi_0\right| X_T^\dagger H^{(2)} X_T \left|\psi_0 \right\rangle / E_{\mathrm{min}}\\
    \mathrm{s.t.}\quad %
                    & H_k = u_{1k}H^{(1)} + u_{2k}H^{(2)},\ k=1,\ldots,T\\
                    \label{eq:model-energy-h1}
                    & H^{(1)} = -\sum_{i=1}^q \sigma_i^x
                    ,\ H^{(2)} = \sum_{ij} J_{ij}\sigma_i^z \sigma_j^z\\
                    & X_{k}=e^{-i H_k \Delta t}X_{k-1},\ k=1,\ldots,T \nonumber \\
    & X_0 = X_\mathrm{init}\nonumber \\
                    & u_{1k} + u_{2k} =1,\ k=1,\ldots,T \\
                    & u_{1k},\ u_{2k} \in \left\{0,1\right\},\ k=1,\ldots,T,
\end{align}
\end{subequations}
where $\sigma_i^x,\ \sigma_i^z$ are Pauli matrices of qubit $i$ for $i=1,\ldots,q$. The matrix $[J_{ij}],\ i,j=1,\ldots,q$ is the adjacency matrix of a randomly generated graph with $q$ nodes. Specifically, $J$ is a symmetric matrix with zero diagonals and other elements randomly generated uniformly from the range $[-1,1]$ for $q> 2$. When $q=2$, $J$ is a symmetric matrix with zero diagonals and other elements as $1$.

We assume that $\sigma_j^\mathrm{time}=0.1\sigma_j^\mathrm{offset}$ for control Hamiltonians $j=1,\ 2$. We set the CVaR risk-level parameter $\eta=0.05$, the number of qubits $q=6$, the evolution time $t_f=5$, the number of time steps for solving the continuous relaxation $T=50$, and the number of time steps for rounding $T_R=200$.
We conduct out-of-sample tests for all controls under the same distribution as in-sample tests across $10$ groups, each with $500$ scenarios. We present the results of various numbers given by different scenario numbers, weight choices, and variance settings in \Crefrange{sec:res-energy-scenario}{sec:res-energy-variance}, respectively, and then discuss the CPU time of solving the problem with different sizes in \Cref{sec:res-energy-time}.

\subsubsection{Results of Scenarios}
\label{sec:res-energy-scenario}
We set the weight parameter $\alpha=0.5$, variance $\sigma_j^\mathrm{offset}=0.05$ for both controllers $j=1,\ 2$ and solve the stochastic optimization model with in-sample scenarios $S=1,\ 20,\ 100,\ 300,\ 500$.
\Cref{tab:energy-scenario} presents mean values, CVaR function values defined in~\eqref{eq:cvar-sample}, and weighted summation values of the mean and CVaR with weight $\alpha=0.5$.
The ``In-sample objective'' columns represent the results among in-sample scenarios generated to solve the model.
The ``Out-of-sample objective'' columns represent the results across $5,000$ independently generated samples for evaluating different control solutions.
The columns under ``Gap'' represent the gaps between in-sample and out-of-sample tests for all the function values.
\begin{table}[htbp]
  \centering
  \caption{Objective value for in-sample and out-of-sample tests and their gaps with a different number of scenarios, including mean values (``Mean''), CVaR function values (``CVaR''), and weighted summation (``Total'').}
    \begin{tabular}{r|rrr|rrr|rrr}
    \hline
    \multicolumn{1}{c|}{\multirow{2}[0]{*}{$S$}} & \multicolumn{3}{c|}{In-sample objective} & \multicolumn{3}{c|}{Out-of-sample objective} & \multicolumn{3}{c}{Gap} \\
    \cline{2-10}
          & Mean & CVaR  & Total & Mean & CVaR  & Total & Mean & CVaR  & Total \\
          \hline
    1     & 0.038 & 0.038 & 0.038 & 0.107 & 0.357 & 0.232 & 64.16\% & 89.27\% & 83.48\% \\
    20    & 0.100 & 0.206 & 0.153 & 0.141 & 0.410 & 0.276 & 29.07\% & 49.82\% & 44.50\% \\
    100   & 0.105 & 0.290 & 0.198 & 0.104 & 0.326 & 0.215 & $-$0.84\% & 10.96\% & 8.10\% \\
    300   & 0.102 & 0.286 & 0.194 & 0.108 & 0.309 & 0.208 & 5.67\% & 7.29\% & 6.87\% \\
    500   & 0.100 & 0.300 & 0.200 & 0.099 & 0.318 & 0.208 & $-$0.90\% & 5.64\% & 4.08\% \\
    \hline
    \end{tabular}%
  \label{tab:energy-scenario}%
\end{table}%
We show that generally the gap of all the objective values decreases when the number of scenarios increases and the gap of the CVaR function value is higher than the mean value.

\subsubsection{Results of Weight Parameter} \label{sec:res-energy-weight}
For the remaining tests we fix the in-sample size as $S=300$ and keep the out-of-sample size as 5,000 scenarios.
 We set the variance $\sigma_j^\mathrm{offset}=0.05$ for $j=1,2$ and solve the stochastic optimization model with different weight parameters $\alpha=0,\ 0.25,\ 0.5,\ 0.75,\ 1$.
 When $\alpha=0$, the model optimizes only the CVaR function; in comparison, when $\alpha=1$, the model optimizes only the expected value of the random objective.
\Cref{fig:energy-weight} presents how the values of mean and CVaR in out-of-sample tests vary depending on the weight parameter.
The blue line marked by dots represents the mean, and the orange line marked by triangles represents the CVaR.
Furthermore, we present box plots describing the objective values of 5,000 out-of-sample test scenarios for each weight parameter $\alpha$. The red lines, the box edges, and the caps represent the medians, the first to the third quartiles, and the whiskers based on the interquartile range, respectively (see~\citet{wickham201140} for details).

\begin{figure}[ht]
    \centering
    \includegraphics[width=0.55\textwidth]{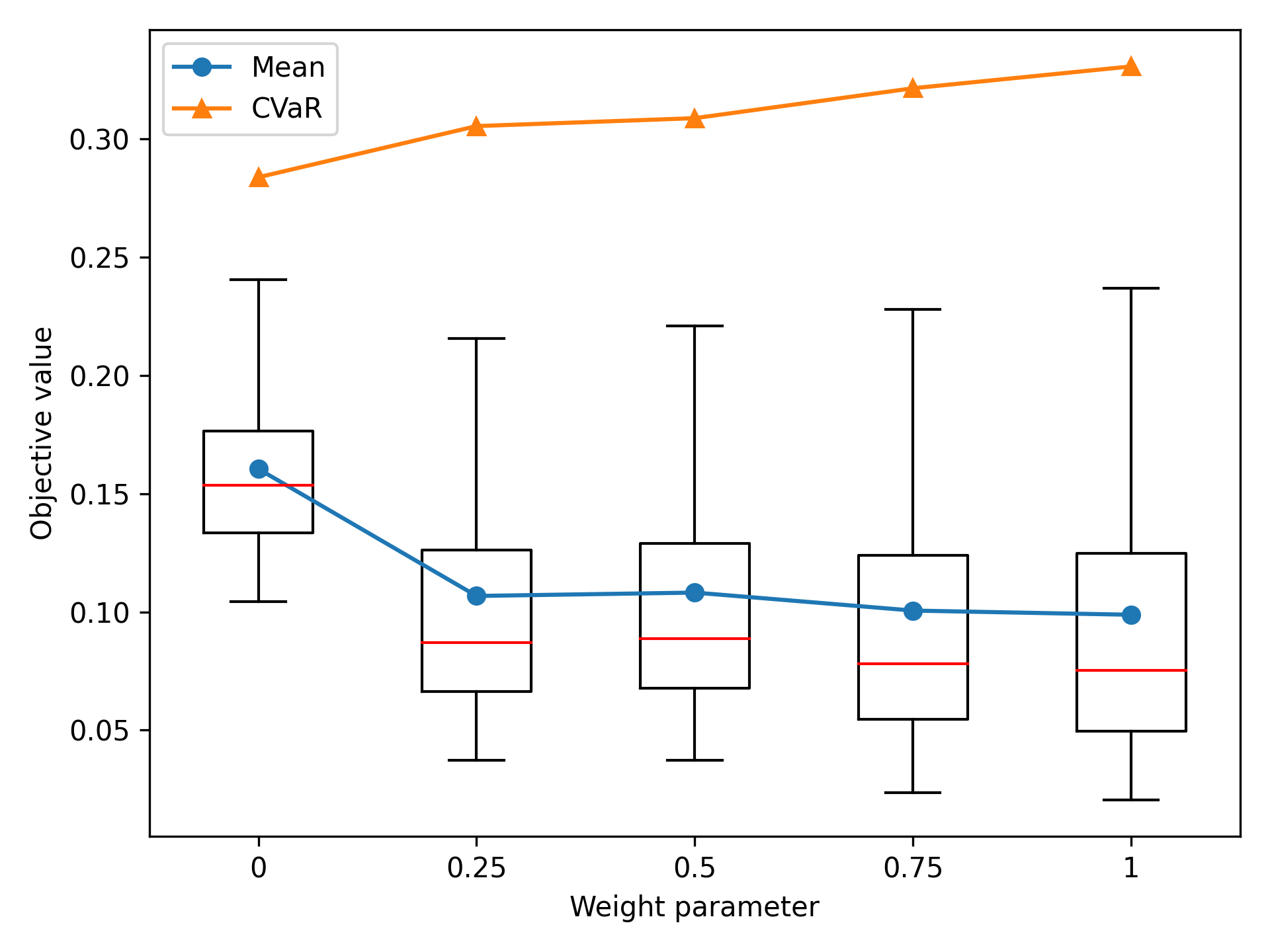}
    \caption{Objective values in out-of-sample tests with various weight parameters $\alpha$. The blue line marked by dots represents the mean value. The orange line marked by triangles represents the CVaR function value. Red lines, box edges, and caps represent medians, first and third quartiles, and whiskers~\citep{wickham201140}, respectively.}
    \label{fig:energy-weight}
\end{figure}

We see that when $\alpha$ increases, the out-of-sample mean values decrease while the CVaR function values increase because the objective function assigns more weight to the expectation.
Moreover, the box plots illustrate that decreasing $\alpha$ results in a reduced deviation, showing the advantages of incorporating risk aversion into the objective function.

In our out-of-sample tests we find that with the same offset $\mu^s_j$ for $s=1,\ldots,S$, the standard deviations of the out-of-sample objective value with different $\xi^s$ are always smaller than $0.005$.
Therefore, we focus on comparing the objective values with various offsets in our following discussion.
Using a derived control $u$ from a given weight parameter $\alpha$ and offsets $\mu_1,\ \mu_2\in [-0.5,0.5]$, we generate $20$ scenarios for $\xi$ following the normal distribution $\mathcal{N}(\mu_j, \sigma_j^{\mathrm{time}}),\ j=1,2$ and compute the objective value $F_X(X_T(u; \xi))$.
The average objective value is considered as the performance of control $u$ under a specific simulation uncertainty offset $(\mu_1,\mu_2)$.
In \Cref{fig:energy-func-noise} we select the risk-averse case ($\alpha=0$) and the risk-neutral case ($\alpha=1$) to present the figures of average objective value among 20 scenarios for offsets $\mu_1,\ \mu_2\in [-0.5,0.5]$.

\begin{figure}[ht]
    \centering
    \subfloat[Risk-averse ($\alpha=0$)]{\includegraphics[width=0.5\textwidth]{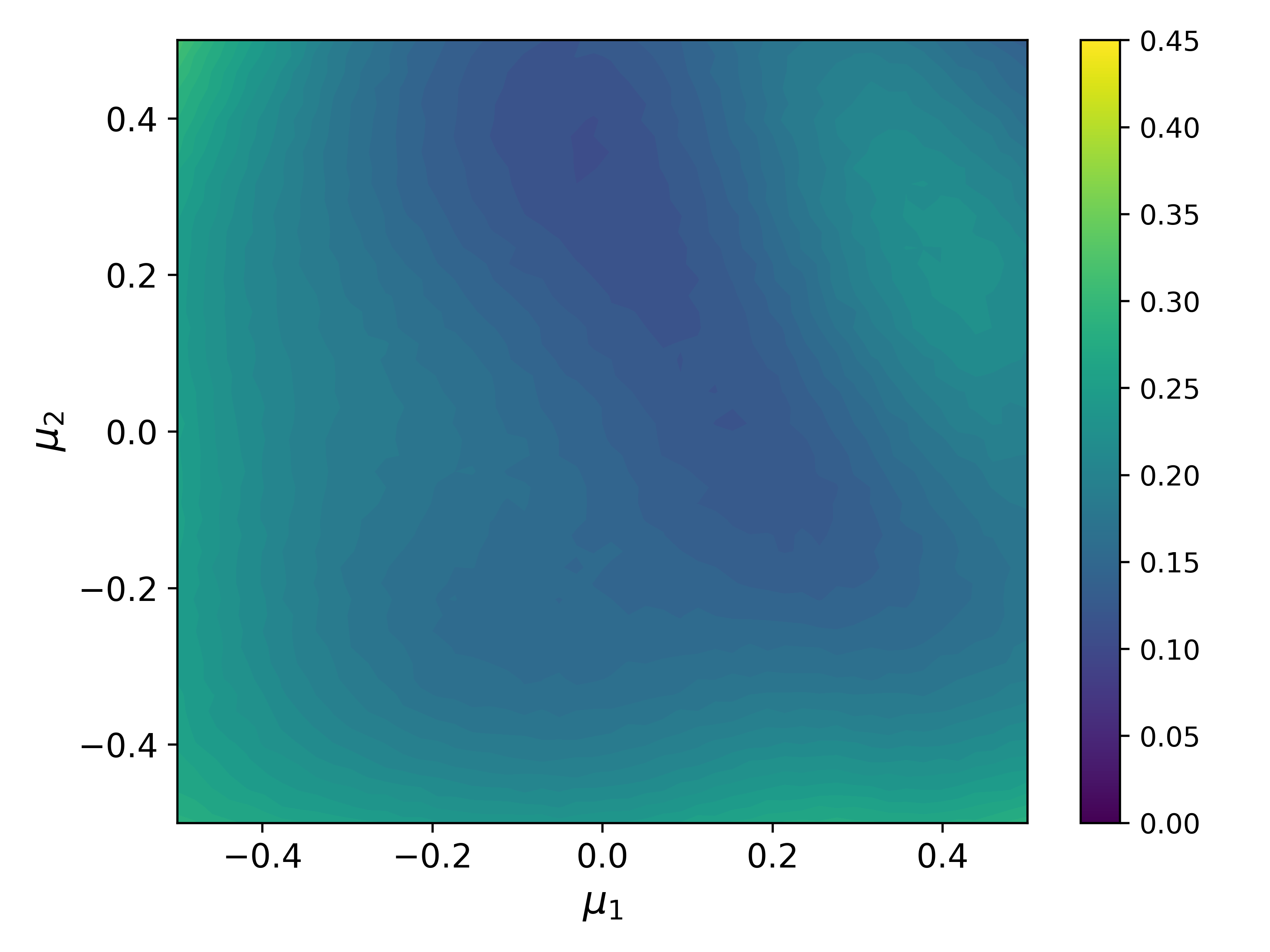}}
    \subfloat[Risk-neutral ($\alpha=1$)]{\includegraphics[width=0.5\textwidth]{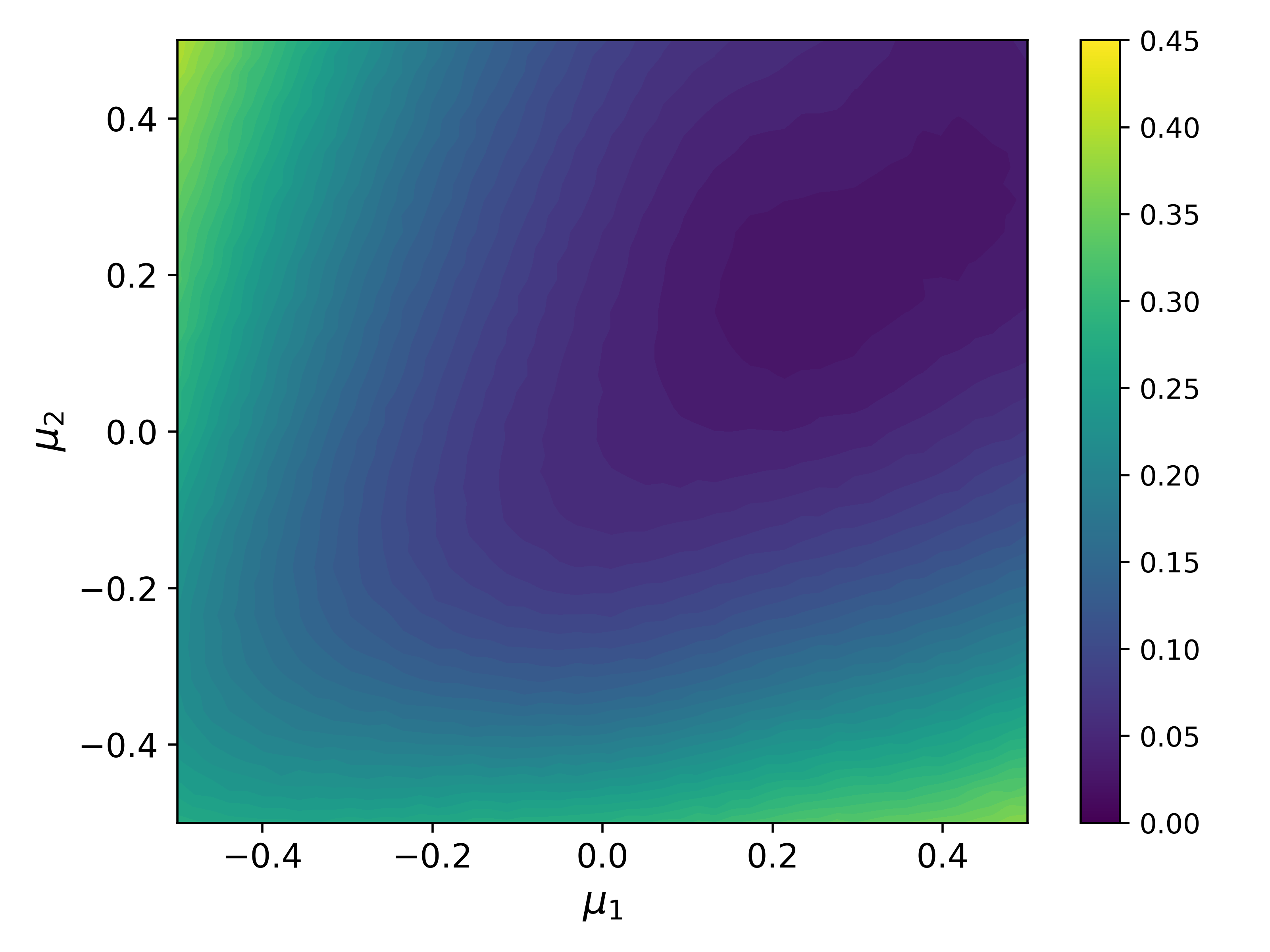}\label{fig:energy-func-noise-sp}}
    \caption{Average objective values among samples of uncertainty $\xi$ as a function of uncertainty offsets $\mu_1,\ \mu_2\in [-0.5,0.5]$.
    The control solutions are obtained from the stochastic optimization model with $\alpha=0,\ 1$ and variance as $0.05$.}
    \label{fig:energy-func-noise}
\end{figure}
We see that the control of the risk-neutral case attains a lower objective value when the uncertainty offsets are small, leading to better average performance. On the other hand, the control of the risk-neutral case has a significantly higher objective value when the uncertainty offsets are large, as shown in the upper-left and lower-right corners of \Cref{fig:energy-func-noise-sp}, while the control of the risk-averse case is more robust among all scenarios.

\subsubsection{Results of Variance}
\label{sec:res-energy-variance}

Again, we fix the number of scenarios $S=300$ and solve the stochastic optimization model with weight $\alpha=0,\ 0.5,\ 1$ for different choices of variances as $\sigma_1^\mathrm{offset},\ \sigma_2^\mathrm{offset}\in \{0.01,\ 0.05,\ 0.1\}$.
We evaluate the derived control solution using three metrics: the mean value, the CVaR function value, and the success rate in distinguishing states in the energy minimization problem. In this instance, a control successfully distinguishes the first excited state from the minimum energy state if its objective value is smaller than the energy difference ratio of these two states, which is one of our control design goals. The third metric is thus the percentage of scenarios in the out-of-sample tests that achieve this distinction, defined as the distinguished percentage (DP).

To more straightforwardly compare the performance of the stochastic optimization models with the deterministic model, we compute the percentage of change in metrics as $(me_{SP}-me_{D})/me_{D}$, where $me_{D}$ and $ me_{SP}$ represent the evaluation metric value of the deterministic model and stochastic optimization model, respectively.
Our goal is to achieve a lower quantum system energy and a higher success rate in distinguishing states. Therefore, for mean and CVaR, a negative percentage of change means a lower energy consumption and better performance compared with the deterministic model; in contrast,
for the DP, a positive percentage of change means a higher success rate of the distinction and better performance compared with the deterministic model.
We present the percentage of change in \Cref{tab:energy-var} and bold the best results of each variance setting.
Columns with ``$\alpha=1$,'' ``$\alpha=0$,'' and ``$\alpha=0.5$'' represent the results of the model optimizing the expectation, optimizing the CVaR function, and optimizing the weighted summation function, respectively.

For the mean value we show that the results with $\alpha=1$ are always better performing (in our tests) compared with the deterministic model, while the results with $\alpha=0$ have worse performance because the stochastic model focuses only on the tail distribution.
The balanced model with $\alpha=0.5$ performs worse with low variance but better with high variance.
For the CVaR function value we show that models with all the weights have better results and the model with $\alpha=0$ is the best, demonstrating an improvement in robustness when considering parameter uncertainty.
For the DP we show that the models with $\alpha=1$ and $\alpha=0.5$ are both
better than the deterministic model for all the variance settings, showing the
benefits of our stochastic optimization model. The model with $\alpha=0$
performs worse with high variance because it optimizes for scenarios with a
high error $\xi$ and sacrifices the performance in other scenarios.

\begin{table}[htbp]
  \centering
  \caption{Percentage change compared with the deterministic model in mean values (``Mean''), CVaR function values (``CVaR''), and percentage of successfully distinguishing first excited state of different offset variances. The results include the model optimizing mean (``$\alpha=0$''), CVaR function (``$\alpha=1$''), and weighted summation of two functions (``$\alpha=0.5$''). The in-sample and out-of-sample tests have the same distribution. The best results are bolded. }
  \begin{adjustbox}{width=\textwidth}
    \begin{tabular}{rr|rrr|rrr|rrrrr}
    \hline
    \multicolumn{1}{c}{\multirow{2}[0]{*}{$\sigma_1^\mathrm{offset}$}} & \multicolumn{1}{c|}{\multirow{2}[0]{*}{$\sigma_2^\mathrm{offset}$}} & \multicolumn{3}{c|}{Mean} & \multicolumn{3}{c|}{CVaR} & \multicolumn{3}{c}{DP} \\
   \cline{3-11}
    &       & \multicolumn{1}{r}{$\alpha=1$} & \multicolumn{1}{r}{$\alpha=0$} & \multicolumn{1}{r|}{$\alpha=0.5$} & \multicolumn{1}{r}{$\alpha=1$} & \multicolumn{1}{r}{$\alpha=0$} & \multicolumn{1}{r|}{$\alpha=0.5$} & \multicolumn{1}{r}{$\alpha=1$} & \multicolumn{1}{r}{$\alpha=0$} & \multicolumn{1}{r}{$\alpha=0.5$} \\
    \hline
    0.01  & 0.01  & \textbf{$-$3.43\%} & 5.29\% & $-$1.15\% & $-$4.10\% & \textbf{$-$5.24\%} & $-$1.92\% & \textbf{0.00\%} & \textbf{0.00\%} & \textbf{0.00\%} \\
    0.01  & 0.05  & \textbf{$-$5.30\%} & 38.85\% & 0.22\% & $-$4.90\% & \textbf{$-$17.55\%} & $-$14.95\% & 0.25\% & \textbf{1.44\%} & 1.03\% \\
    0.01  & 0.1   & \textbf{$-$8.85\%} & 80.32\% & 4.43\% & $-$7.21\% & \textbf{$-$26.08\%} & $-$17.66\% & 1.68\% & $-$9.27\% & \textbf{4.20\%} \\
    0.05  & 0.01  & \textbf{$-$8.65\%} & 25.35\% & 4.28\% & $-$9.30\% & \textbf{$-$29.15\%} & $-$16.33\% & 0.95\% & \textbf{2.57\%} & 1.53\% \\
    0.05  & 0.05  & \textbf{$-$9.06\%} & 47.67\% & $-$0.44\% & $-$7.50\% & \textbf{$-$20.61\%} & $-$13.62\% & 2.04\% & \textbf{3.08\%} & 2.89\% \\
    0.05  & 0.1   & \textbf{$-$11.45\%} & 64.49\% & $-$7.61\% & $-$7.29\% & \textbf{$-$19.86\%} & $-$12.18\% & 3.71\% & $-$24.40\% & \textbf{4.20\%} \\
    0.1   & 0.01  & \textbf{$-$11.66\%} & 38.76\% & $-$2.02\% & $-$8.96\% & \textbf{$-$31.38\%} & $-$18.00\% & 3.87\% & 7.97\% & \textbf{6.00\%} \\
    0.1   & 0.05  & \textbf{$-$12.12\%} & 33.77\% & $-$5.13\% & $-$7.87\% & \textbf{$-$23.16\%} & $-$12.37\% & 5.20\% & 0.74\% & \textbf{6.38\%} \\
    0.1   & 0.1   & \textbf{$-$14.26\%} & 55.23\% & $-$5.45\% & $-$6.87\% & \textbf{$-$17.99\%} & $-$10.52\% & 6.79\% & $-$72.22\% & \textbf{7.08\%} \\
    \hline
    \end{tabular}%
  \label{tab:energy-var}%
  \end{adjustbox}
\end{table}%

Furthermore, we observe that with a fixed uncertainty offset variance for one controller, increasing the variance of the other controller leads to higher mean values, higher CVaR function values, and lower DP,
because the uncertainty in the quantum system increases.
Increasing the uncertainty offset variance of the first controller $H^{(1)}$ has a larger negative impact on objective values compared with increasing the uncertainty variance of $H^{(2)}$, which means this quantum control system is more sensitive to the uncertainty of controller $H^{(1)}$.

We show the histogram of out-of-sample tests for both deterministic and stochastic optimization models and the zoomed-in tail distribution with a variance of $0.05$ in \Cref{fig:energy-var-histogram}. The blue and yellow histograms represent the results of the deterministic and the stochastic optimization models, respectively. The figures show that our stochastic optimization model obtains a lighter tail distribution.

\begin{figure}[ht]
    \centering
    \subfloat[All scenarios]{\includegraphics[width=0.5\textwidth]{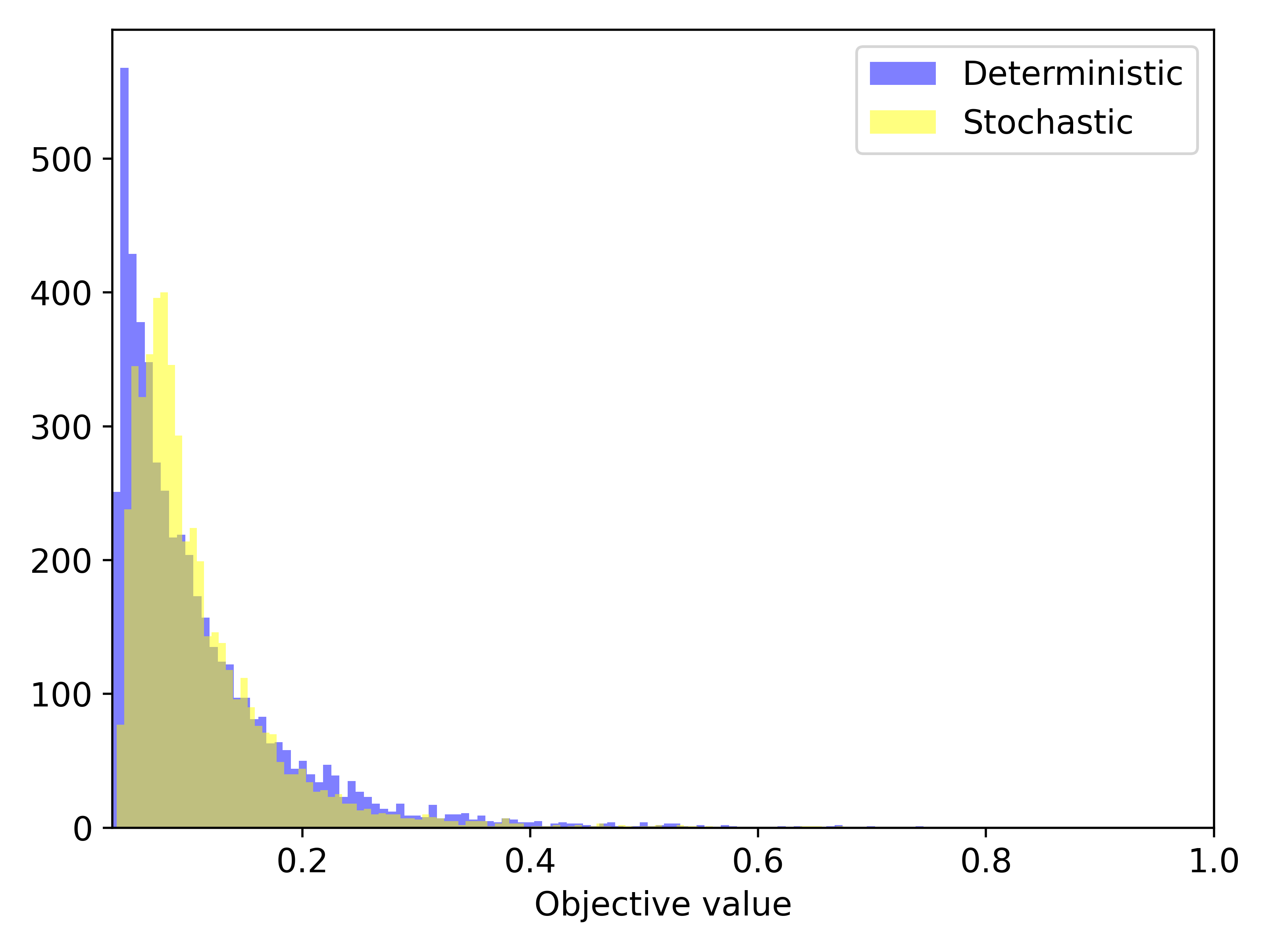}}
    \subfloat[Tail distribution]{\includegraphics[width=0.5\textwidth]{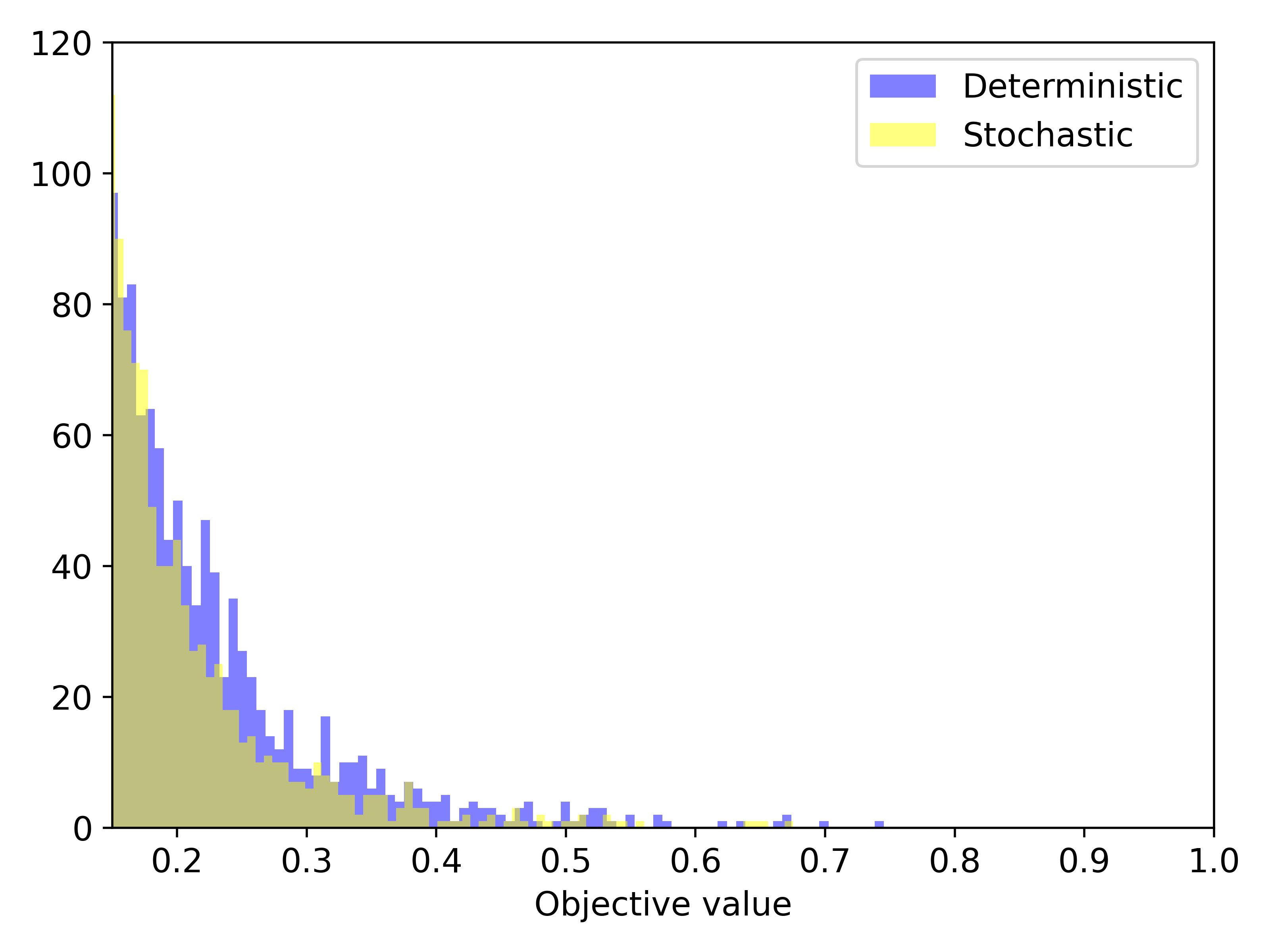}}
    \caption{Histograms of out-of-sample tests for the deterministic and stochastic optimization model with offset variance $0.05$ for both controllers. Blue and yellow histograms represent the results of the deterministic and the stochastic optimization model, respectively. %
    }
    \label{fig:energy-var-histogram}
\end{figure}

Similar to \Cref{fig:energy-func-noise}, for each obtained control $u$ and for every combination of offsets value $\mu_1,\ \mu_2\in [-1, 1]$, we generate $20$ different scenarios for $\xi$ with a normal distribution $\mathcal{N}(\mu_j,\sigma_j^\mathrm{time}),\ j=1,2$ and compute the average objective value $F_X(X_T(u; \xi))$.
The average objective value represents the performance of control $u$ under a specific simulation uncertainty offset $(\mu_1,\mu_2)$.
In \Cref{fig:energy-func-noise-dvssp} we present the average objective values for different offset values $\mu_1,\ \mu_2\in [-1,1]$ for $u$ obtained from the deterministic and stochastic optimization models with both offset variances set as $0.1$.
We show that although both controls have high objective value when $|\mu|$ goes to $1$, the control of the stochastic optimization model is more robust, especially for $(\mu_1,\ \mu_2)\in [-1,-0.75]\times[0.5,1]$ and $(\mu_1,\ \mu_2)\in [0.5,1]\times[-1,-0.75]$.

\begin{figure}[ht]
    \centering
    \subfloat[Deterministic]{\includegraphics[width=0.5\textwidth]{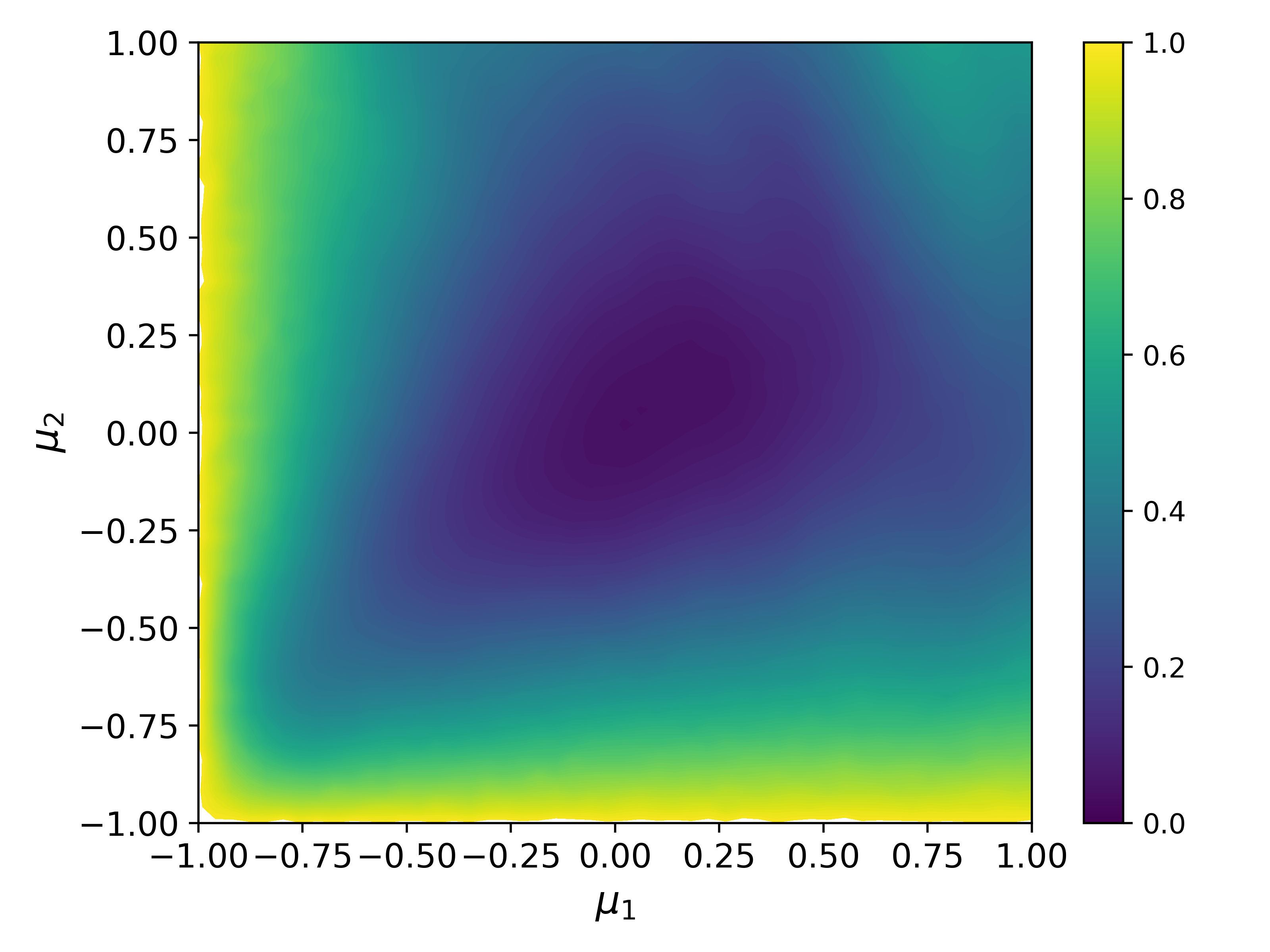}}
    \subfloat[Stochastic]{\includegraphics[width=0.5\textwidth]{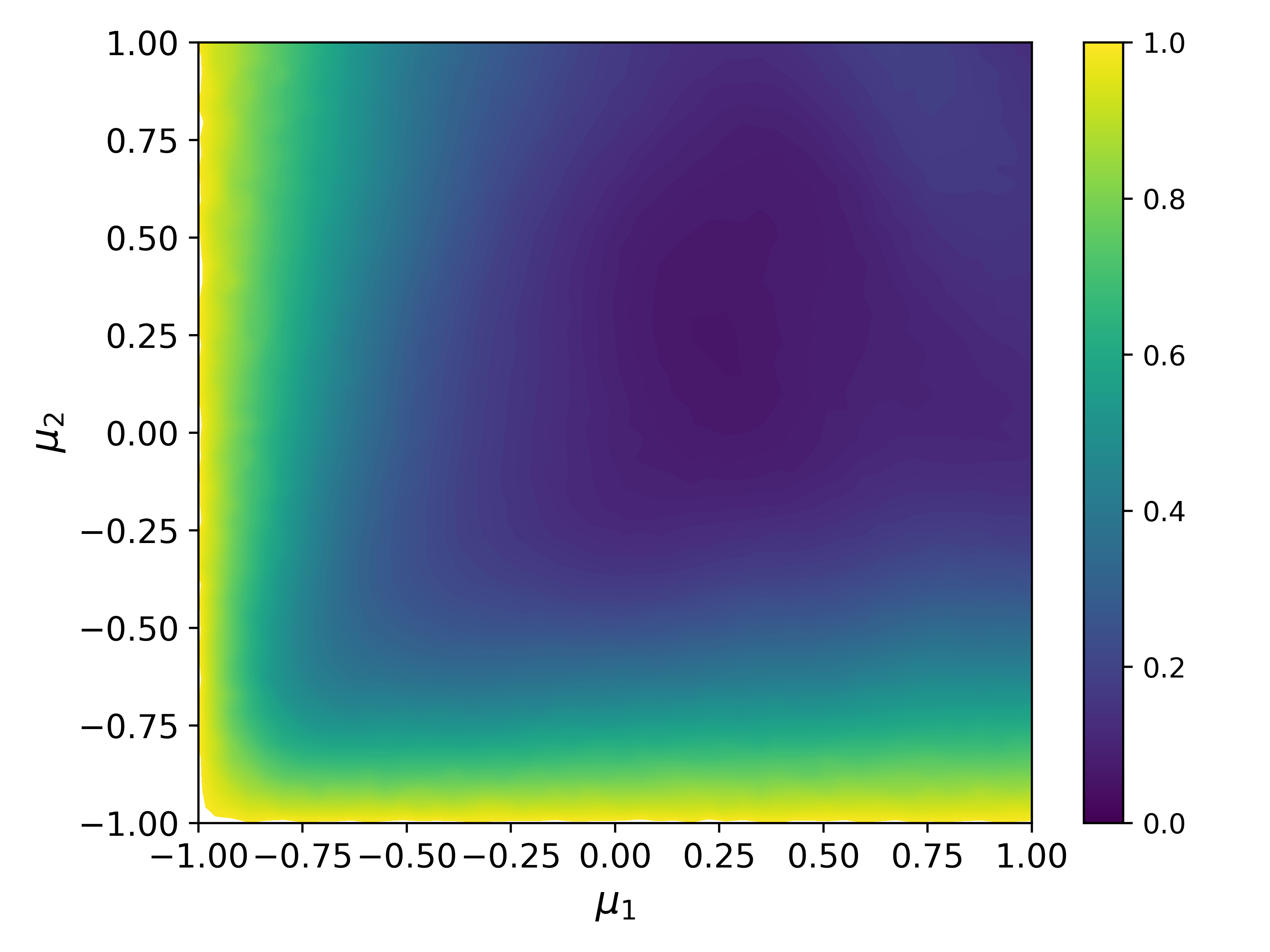}}
    \caption{Average objective values among samples of uncertainty $\xi$ as a function of $\mu_1,\ \mu_2\in [-1,1]$. The control solutions are obtained from the deterministic and stochastic optimization model with $\alpha=0.5$ and offset variances as $0.1$.}
    \label{fig:energy-func-noise-dvssp}
\end{figure}

With a given risk level parameter $\eta=0.01,\ 0.05,\ 0.1$, we solve the in-sample stochastic optimization model with weight $\alpha=0$ and obtain the $1-\eta$ percentile value among all the scenarios.
To evaluate the performance of our model on controlling the risk, we present the percentage of scenarios in the out-of-sample tests with an objective value smaller than the $1-\eta$ percentile of the in-sample objective value distribution in \Cref{tab:res-cvar}.
We show that for all the variance settings, the difference between the percentage and the ideal value $1-\eta$ is mostly smaller than 1\%.
We have the closest percentage when the risk level is $\eta=0.05$. When the risk control is too strict or too relaxed, the percentage in out-of-sample tests is usually smaller than $1-\eta$, and the difference increases when variances increase.
\begin{table}[htbp]
  \centering
  \caption{Percentage of scenarios in out-of-sample tests with an objective value smaller than $1-\eta$ percentile of the in-sample distribution.}
    \begin{tabular}{rr|rrr}
    \hline
    $\sigma_1^\mathrm{offset}$ & {$\sigma_2^\mathrm{offset}$} & $1-\eta=99\%$ & $1-\eta=95\%$ & {$1-\eta=90\%$} \\
    \hline
    0.05  & 0.05  & 98.68\% & 95.76\% & 89.82\% \\
    0.05  & 0.1   & 98.54\% & 95.68\% & 90.34\% \\
    0.1   & 0.05  & 98.62\% & 95.46\% & 87.74\% \\
    0.1   & 0.1   & 98.26\% & 95.50\% & 88.42\% \\
    \hline
    \end{tabular}%
  \label{tab:res-cvar}%
\end{table}%

\subsubsection{Results of CPU Time}
\label{sec:res-energy-time}
We fix the offset variances $\sigma_j^\mathrm{offset}=0.05,\ j=1,\ 2$, and  the weight parameter $\alpha=0.5$ and solve the stochastic optimization model with different numbers of qubits $q=2,\ 6$, different numbers of time steps $T=20,\ 50$, and different numbers of scenarios $S=1,\ 20,\ 100,\ 200,\ 300$.
We present the CPU time and the number of iterations for L-BFGS-B in \Cref{tab:energy-time}.
\begin{table}[ht]
  \centering
  \caption{CPU time and iteration results of different problem sizes, including the number of qubits $q$, the number of time steps $T$, and the number of scenarios $S$.}
    \begin{tabular}{rrrrr}
    \hline
    {$q$} & \multicolumn{1}{r}{$T$} & \multicolumn{1}{r}{$S$} & \multicolumn{1}{r}{CPU time (s)} & \multicolumn{1}{r}{Iteration} \\
    \hline
    2     & 20    & 300   & 32.90 & 15 \\
    2     & 50    & 300   & 99.74 & 27 \\
    6     & 50    & 300   & 2814.33 & 26 \\
    6     & 50    & 200   & 971.06 & 14 \\
    6     & 50    & 100   & 401.20 & 19 \\
    6     & 50    & 20    & 121.18 & 15 \\
    6     & 50    & 1     & 21.66 & 34 \\
    \hline
    \end{tabular}%
  \label{tab:energy-time}%
\end{table}%
We show that the number of qubits $q$ has the most important impact on the CPU time because the dimension of Hamiltonian matrices grows exponentially with $q$.
This issue can be potentially resolved in the future by using quantum computers to conduct time evolution.
An increasing number of scenarios $S$ leads to an increase in CPU time, which can be reduced by parallel computing on multiple CPU cores of classical computers or multiple quantum computers.
The CPU time also increases with the increase in the number of time steps $T$. Moreover, we notice that the number of iterations is robust regardless of the problem size.

\subsection{Circuit Compilation Problem}
\label{sec:res-circuit}
Quantum circuit compilation aims to represent a circuit by specific controllers and constraints, to build a foundation for general quantum algorithms. In this section, we apply the modified Adam method (\Cref{alg:adam}) to study a compilation problem for the quantum circuit that has the ground state energy of molecules
generated by the unitary coupled-cluster single-double method~\citep{bartlett2007coupled, romero2018strategies}.
We consider a gmon qubit quantum system with $q$ qubits, which is a superconducting qubit architecture combining high-coherence qubits and tunable qubit couplings~\citep{chen2014qubit}.
Each qubit has a flux-drive controller and a charge-drive controller, and they are connected with their nearest neighbors according to a rectangular-grid topology.
The set of connected qubits is denoted by $E$ with size $|E|$. Each connected qubit group in $E$ has a corresponding Hamiltonian controller.
The initial operator $X_\mathrm{init}$ is a $2^q$-dimensional identity matrix, and the target operator $X_\mathrm{targ}$ is the matrix formulation of the circuit for a certain molecule. The specific formulation of the deterministic model is
\begin{subequations}
\label{eq:model-molecule-d}
\begin{align}
    \min\quad  & 1 - \frac{1}{2^q}\left|\operatorname{tr}\left\{X_\mathrm{targ}^\dagger X_T\right \}\right| \\
    \mathrm{s.t.}\quad  & H_k = \sum_{j=1}^{2q} u_{jk}H^{(j)} + \sum_{(j_1,j_2)\in E} u_{j_1j_2k}H^{(j_1j_2)},\ k=1,\ldots,T\\
                    \label{eq:model-molecule-h1-1}
                    & H^{(2j-1)} = J_c \sigma^x_{j},\ H^{(2j)} = J_f \left(\begin{array}{ll}
                                0 & 0 \\
                                0 & 1
                                \end{array}\right)_{(j)},\ j=1,\ldots,q\\
                    \label{eq:model-molecule-hc}
                    & H^{(j_1j_2)} = J_e \sigma^x_{j_1}\sigma^x_{j_2},\ \forall (j_1, j_2)\in E\\
                    & X_{k}=e^{-i H_k \Delta t}X_{k-1},\ k=1,\ldots,T \nonumber \\
                    & X_0 = X_\mathrm{init}\nonumber \\
                    & \sum_{j=1}^{2q} u_{jk} + \sum_{(j_1,j_2)\in E} u_{j_1j_2k}=1,\ k=1,\ldots,T\\
                    & u_{jk},\ u_{j_1j_2k} \in \left\{0,1\right\},\ j=1,\ldots,2q,\ \forall (j_1,j_2)\in E,\ k=1,\ldots,T,
\end{align}
\end{subequations}
where $\sigma_j^x$ are Pauli matrices for qubits $j=1,\ldots,q$ and the subscript $(j)$ in constraints~\eqref{eq:model-molecule-h1-1} represents the matrix operation acting on the $j$th qubit. The constants $J_c,\ J_f,\ J_e$ correspond to specific quantum machines and are set as $J_c = 0.2\pi,\ J_f=3\pi,\ J_e=0.1\pi$.

We assume that the variance of uncertainty among time steps $\sigma_j^\mathrm{time}=0.1\sigma_j^\mathrm{offset}$ for all the control Hamiltonians. All the single-qubit control Hamiltonians have the same uncertainty offset variance, represented by $\sigma_s^{\mathrm{offset}}$; and all two-qubit control Hamiltonians have the same uncertainty offset variance, represented by $\sigma_t^\mathrm{offset}$.
In \Crefrange{sec:res-circuit-scenario}{sec:res-circuit-variance}
we discuss the performance of the stochastic optimization model on an instance of the molecule H$_2$ (dihydrogen). The system includes $q=2$ qubits, $4$ single-qubit controllers, and a two-qubit controller. We set the evolution time $t_f=20$, number of time steps $T=50$, and number of rounding time steps $T_R=4000$; the risk level $\eta=0.05$.
We generate 10 groups, each with 500 scenarios sampled from the same distribution under in-sample tests, to conduct out-of-sample tests for evaluating the obtained controls.
In \Cref{sec:res-circuit-time} we present the CPU time of solving the circuit compilation problem with different molecules and problem sizes.

\subsubsection{Results of Scenarios}
\label{sec:res-circuit-scenario}
In this section we set the weight parameter $\alpha=0.5$ and offset variances $\sigma_s^\mathrm{offset}=\sigma_t^\mathrm{offset}=0.01$.
We test our algorithm with a different number of scenarios $S=20,\ 40,\ 80$, and $160$ with adjusted learning rates of $0.05,\ 0.06,\ 0.08,$ and $0.15$. To compare the performance under the same computational costs, which is represented by the product of the number of scenarios and iterations ($S\times K$), we set the number of iterations to $2000,\ 1000,\ 500,$ and $250$ accordingly.
In \Cref{fig:h2-scenario} we show how the objective value varies with the computational costs during the algorithm procedure by a log-log scale.
We show that with a larger number of scenarios, the objective value is more stable because the method learns more about the distribution at each iteration. However, the convergence is slower because the algorithm runs for fewer iterations.
\begin{figure}[ht]
    \centering
    \includegraphics[width=0.55\textwidth]{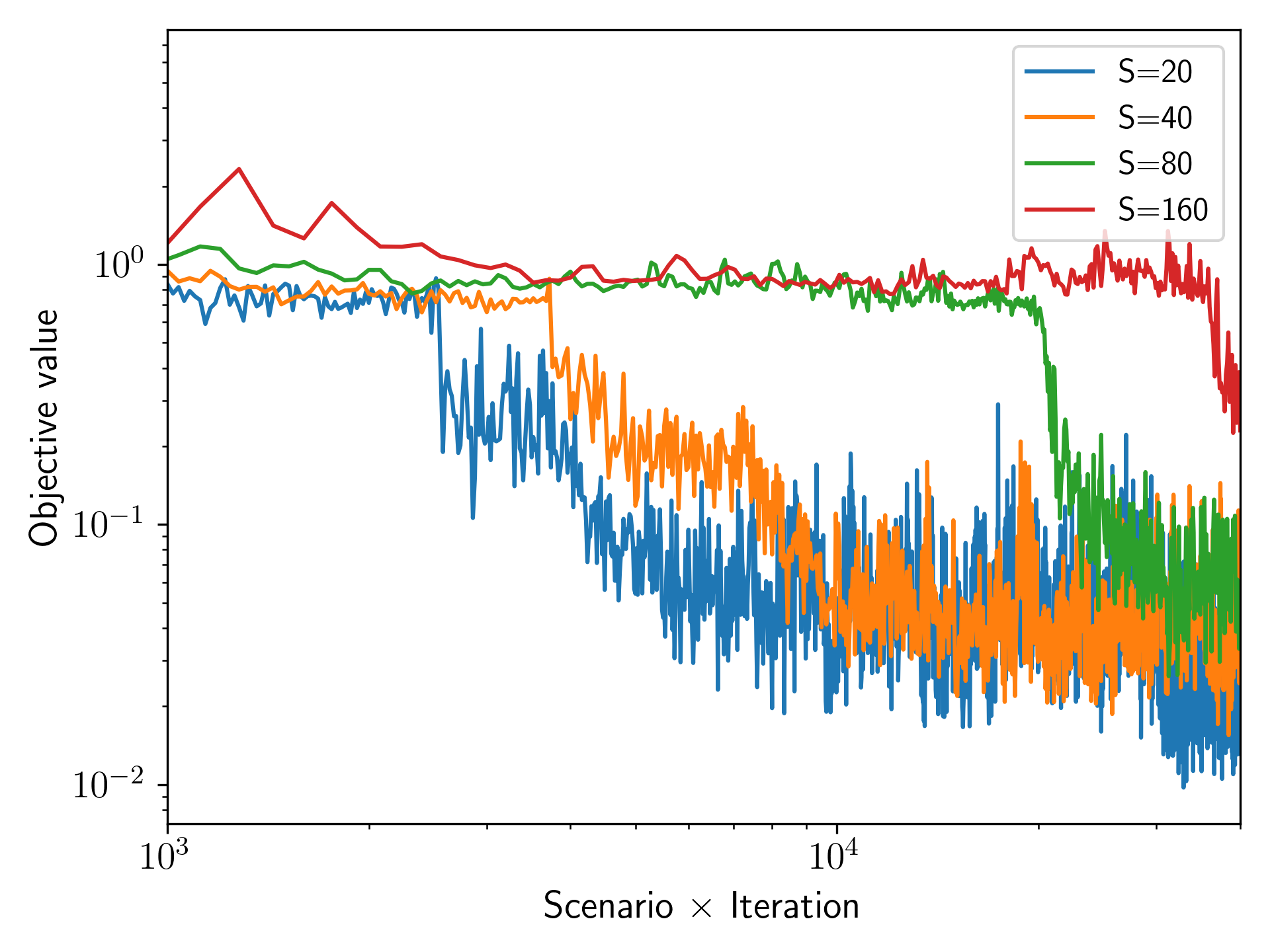}
    \caption{Log-log scale figure for the objective values during the in-sample test iterations. The x-axis represents the multiplication of the number of scenarios and iterations. Blue, orange, green, and red lines represent $S=20$, 40, 80, and 160, respectively.}
    \label{fig:h2-scenario}
\end{figure}

We present the out-of-sample test results for the controls obtained by a different number of scenarios, including the mean value, the CVaR function value, and the total objective value as weighted summation with $\alpha=0.5$ in \Cref{tab:h2-scenario}. We show that the control with $S=20$ achieves the lowest objective value primarily because of its higher number of iterations within the same computational cost.
\begin{table}[ht]
  \centering
  \caption{Mean, CVaR function value, and total objective values ($\alpha=0.5$) in out-of-sample tests for a different number of scenarios.
The offset variances for both in-sample and out-of-sample tests and all the controllers are 0.01.}
    \begin{tabular}{rrrr}
    \hline
    Scenario & Mean & CVaR  & Total \\
    \hline
    20    & 8.19$\times 10^{-3}$ & 3.21$\times 10^{-2}$ & 2.02$\times 10^{-2}$ \\
    40    & 1.21$\times 10^{-2}$ & 4.30$\times 10^{-2}$ & 2.76$\times 10^{-2}$ \\
    80    & 1.35$\times 10^{-2}$ & 6.76$\times 10^{-2}$ & 4.06$\times 10^{-2}$ \\
    160   & 2.38$\times 10^{-2}$ & 9.51$\times 10^{-2}$ & 5.94$\times 10^{-2}$ \\
    \hline
    \end{tabular}%
  \label{tab:h2-scenario}%
\end{table}%

\subsubsection{Results of Variance}
\label{sec:res-circuit-variance}
We compare the performance of the deterministic and the stochastic optimization model with sample size $S=20$, weight parameter $\alpha=0.5$ under different offset variances $\sigma_s^\mathrm{offset}\in \{0.01, 0.05\},\ \sigma_t^\mathrm{offset}\in \{0.01, 0.05\}$.
In \Cref{tab:h2-variance} we present the mean value and the CVaR function value of the deterministic model (represented by ``Deter'') and the stochastic program (represented by ``SP'') for different variances of the uncertainty offsets.

\begin{table}[htbp]
  \centering
  \caption{Mean values and CVaR function values of different offset variances of single-qubit controllers ($\sigma_s^{\mathrm{offset}}$) and two-qubit controllers ($\sigma_t^{\mathrm{offset}}$) for the deterministic model (``Deter'') and the stochastic program (``SP''). The in-sample and out-of-sample tests are under the same distribution. We bold the better results for each variance setting.}
    \begin{tabular}{rrrrrr}
      \hline
    \multicolumn{1}{c}{\multirow{2}[0]{*}{$\sigma^\mathrm{offset}_s$}} & \multicolumn{1}{c}{\multirow{2}[0]{*}{$\sigma^\mathrm{offset}_t$}} & \multicolumn{2}{c}{Mean} & \multicolumn{2}{c}{CVaR} \\
    \cline{3-6}
          &       & \multicolumn{1}{r}{Deter} & \multicolumn{1}{r}{SP} & \multicolumn{1}{r}{Deter} & \multicolumn{1}{r}{SP} \\
          \hline
    0.01  & 0.01  & 0.639 & \textbf{8.19$\mathbf{\times 10^{-3}}$} & 0.986 & \textbf{3.21$\mathbf{\times 10^{-2}}$} \\
    0.01  & 0.05  & 0.639 & \textbf{8.74$\mathbf{\times 10^{-3}}$} & 0.988 & \textbf{4.32$\mathbf{\times 10^{-2}}$} \\
    \hline
    0.05  & 0.01  & 0.748 & \textbf{5.44$\mathbf{\times 10^{-2}}$} & 0.990 & \textbf{0.353} \\
    0.05  & 0.05  & 0.748 & \textbf{9.84$\mathbf{\times 10^{-2}}$} & 0.990 & \textbf{0.419} \\
    \hline
    \end{tabular}%
  \label{tab:h2-variance}%
\end{table}%

Comparing the results of different variances, we show that the uncertainty in single-qubit controllers significantly affects the objective values more than the two-qubit controllers do, mainly because single-qubit controllers are expected to have more impact on unitary operators and they are the majority of controllers in the quantum system.
For example, the instance of H$_2$ includes 4 single-qubit controllers but only 1 two-qubit controller.
Moreover, increasing variance leads to a larger increase in the CVaR function value, indicating a larger negative impact on scenarios with large deviations.

We demonstrate that the control of the deterministic model performs badly even under a small variance, with all the mean values larger than 0.6 and all the CVaR function values larger than 0.9. On the other hand, the control of our stochastic optimization model performs dramatically better on the mean values and CVaR function values for all the settings, illustrating the advantages of our model considering the uncertainty in quantum control systems.

In \Cref{fig:h2-var-histogram} we present the histogram of out-of-sample tests for both deterministic and stochastic optimization models with variances for all the controllers as 0.01 and 0.05.
The blue and yellow histograms represent the results of the deterministic and the stochastic optimization model, respectively. We show that with the increase of variance, both models have heavier tail distribution, but the stochastic optimization model always has a much lighter tail distribution compared with the deterministic model.

\begin{figure}[ht]
    \centering
    \subfloat[$\sigma_s^\mathrm{offset}=\sigma_t^\mathrm{offset}=0.01$]{\includegraphics[width=0.45\textwidth]{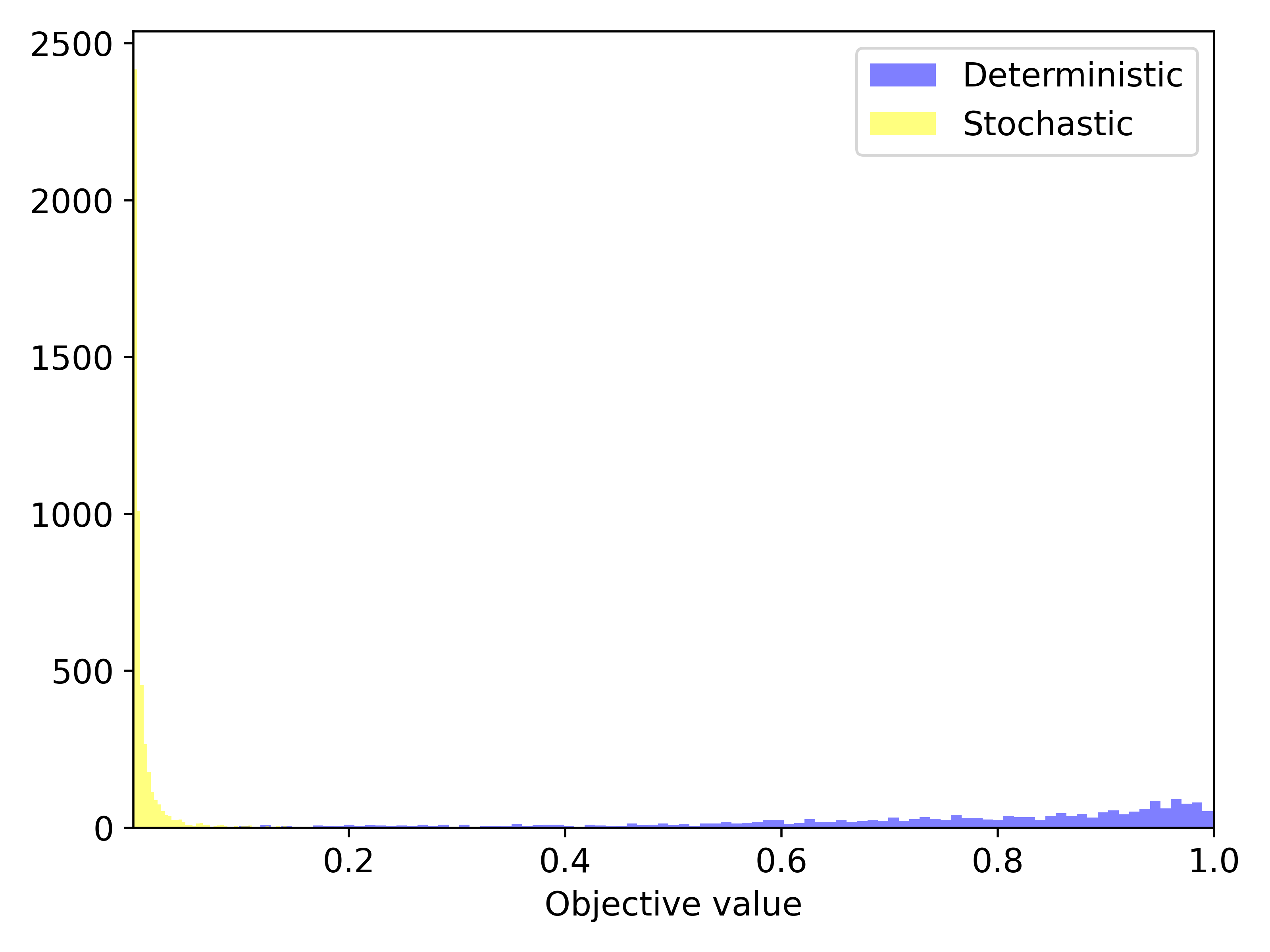}}
    \subfloat[$\sigma_s^\mathrm{offset}=\sigma_t^\mathrm{offset}=0.05$]{\includegraphics[width=0.45\textwidth]{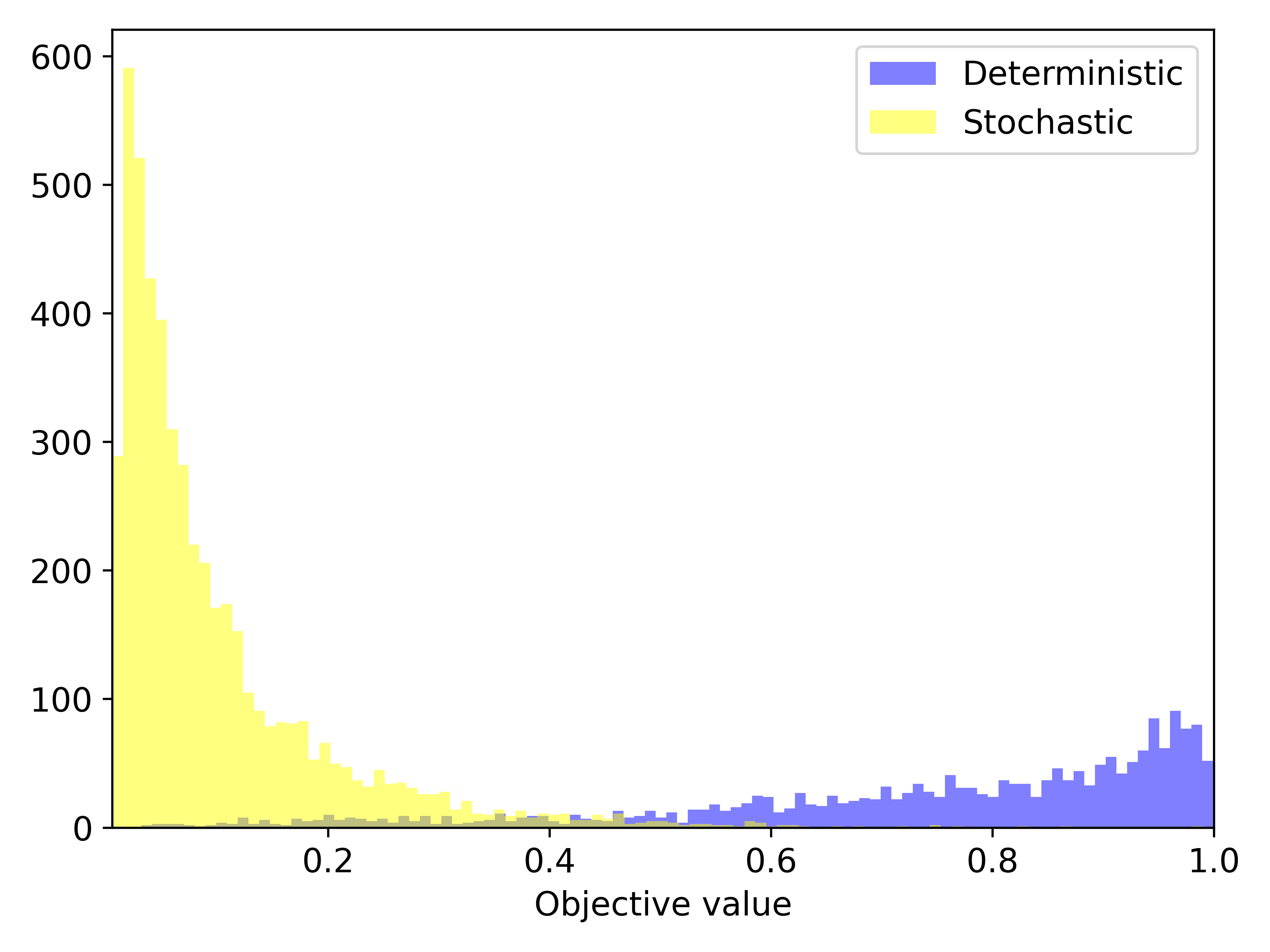}}
    \caption{Histograms of out-of-sample tests with offset variances $0.01$ and 0.05 for all the controllers. Blue and yellow histograms represent the results of the deterministic and stochastic optimization models, respectively.}
    \label{fig:h2-var-histogram}
\end{figure}

\subsubsection{Results of CPU Time}
\label{sec:res-circuit-time}
We set the iteration number for the modified Adam at 2000, weight parameter $\alpha=0.5$, and offset variances $\sigma_s^\mathrm{offset}=\sigma_t^\mathrm{offset}=0.01$.
We solve the stochastic optimization model for molecules H$_2$ and LiH, with time steps $T=50,\ 100$ and scenario numbers $S=20,\ 40$.
In \Cref{tab:h2-time} we present the CPU time of the algorithm for different problem sizes and molecules, with the respective number of qubits $q$ and controllers $N$.
\begin{table}[ht]
  \centering
  \caption{CPU time results of different molecules, different numbers of times steps $T$, and different numbers of scenarios $S$. We present the number of qubits $q$ and the number of controllers $N$ for molecules. }
    \begin{tabular}{crrrrr}
    \hline
    Molecule & $q$     & $N$     & $T$     & $S$     & CPU time (s) \\
    \hline
    &2     & 5     & 50    & 20    & 1383.62 \\
    H$_2$&2     & 5     & 50    & 40    & 2713.80 \\
    &2     & 5     & 100   & 20    & 2888.77 \\
    \hline
    &4     & 12    & 50    & 20    & 4168.16 \\
    LiH&4     & 12    & 50    & 40    & 7918.20 \\
    &4     & 12    & 100   & 20    & 8482.01 \\
    \hline
    \end{tabular}%
  \label{tab:h2-time}%
\end{table}%
We show that with the same number of time steps and scenarios, changing molecules leads to a significant CPU time increase, mainly because the dimension of Hamiltonian matrices increases exponentially with $q$ and the number of controllers also increases. The CPU time increases with time steps $T$ and scenario numbers $S$ approximately linearly. In practice, the CPU time can be potentially reduced by conducting time evolution on quantum computers and parallel computing among different scenarios on large amounts of CPU cores.

\section{Conclusion}
\label{sec:conclusion}
In this paper we built a stochastic mixed-integer program with the sample-based reformulation for the quantum optimal control problem with uncertain Hamiltonians.
We introduced an objective function aiming to balance risk-neutral and risk-averse measurements, which are evaluated by expectation and CVaR function, respectively.
We derived a closed-form expression and discussed the derivative for the objective function.
We modified and applied two gradient-based methods to solve the continuous relaxation and obtained binary solutions by the sum-up rounding technique with a discussion of the rounding errors.

We conducted numerical simulations on multiple quantum control instances.
Based on the results, we recommend L-BFGS-B for quantum control problems minimizing system energy and the modified Adam for problems minimizing infidelity.
The results show that our stochastic optimization model outperforms the deterministic model in terms of both average and robust performance for different variance levels.

With all the simulations completed on classical computers, we find that the number of qubits in quantum systems has a significant impact on the computational time.
Conducting time-evolution processes on quantum computers to reduce computational time is an interesting direction for future research.
Furthermore, model-free optimization methods, including reinforcement learning, provide chances to capture more complex uncertainties in quantum systems.

\section*{Acknowledgements}
This work was supported by the U.S. Department of Energy, Office of Science, Office of Advanced Scientific Computing Research, Accelerated Research for Quantum Computing program  under Contract No. DE-AC02-06CH11357. X. F. and S. S. received partial support from U.S. National Science Foundation (NSF) grant CMMI-2041745 and the Department of Energy (DOE) grant \#DE-SC0018018. L.~T.~B.~is a KBR employee working under the Prime Contract No.~80ARC020D0010 with the NASA Ames Research Center and is grateful for support from the DARPA Quantum Benchmarking program under IAA 8839, Annex 130.

The United States Government retains and the publisher, by accepting the article for publication, acknowledges that the United States Government retains a non-exclusive, paid-up, irrevocable, worldwide license to reproduce, prepare derivative works, distribute copies to the public, and perform publicly and display publicly, or allow others to do so, for United States Government purposes. All other rights are reserved by the copyright owner.

\bibliographystyle{unsrtnat}
\bibliography{Xinyu}
\end{document}